\documentclass[11pt]{article}

\usepackage{amsfonts,amssymb,mathtools}
\usepackage{amsmath} 
\usepackage{hyperref}
\usepackage{cleveref}

\usepackage{graphicx}
\usepackage[caption=false]{subfig}

\usepackage{algorithm} 
\usepackage[noend]{algpseudocode}

\usepackage{enumitem}
\setlist[enumerate]{leftmargin=.5in}
\setlist[itemize]{leftmargin=.5in}

\newtheorem{example}{Example}
\newtheorem{definition}{Definition}
\newtheorem{theorem}{Theorem}
\newtheorem{proof}{Proof}

\title{Predictive refinement methodology for compressed sensing imaging 
}
\author{Alfredo Nava-Tudela\thanks{Institute for Physical Science and Technology, University of Maryland, College Park, MD, USA (\href{mailto:ant@umd.edu}{ant@umd.edu}).}}

\begin{document}
\maketitle

\begin{abstract}
The weak-$\ell^p$ norm can be used to define a measure $s$ of sparsity. When we compute $s$ for the discrete cosine transform coefficients of a signal, the value of $s$ is related to the information content of said signal. We use this value of $s$ to define a reference-free index $\mathcal{E}$, called the sparsity index, that we can use to predict with high accuracy the quality of signal reconstruction in the setting of compressed sensing imaging. That way, when compressed sensing is framed in the context of sampling theory, we can use $\mathcal{E}$ to decide when to further partition the sampling space and increase the sampling rate to optimize the recovery of an image when we use compressed sensing techniques.
\end{abstract}



\section{Introduction} \label{sec: introduction}
In order to reproduce the voice of a singer who can sing up to a ``soprano C", or at a frequency of $W = 1046$ Hz, Claude Shannon \cite{shan1949} proved that we need to sample her voice once every $\frac{1}{2W}$ seconds. He named this number the Nyquist sampling rate for a signal of band $W$ Hz, i.e., a signal with frequencies no higher than $W$ Hz, after Harry Nyquist, who had ``pointed out the fundamental importance of the time interval $\frac{1}{2W}$ in connection with telegraphy."

Shannon notes that this result was known in other forms by the mathematician J.\ M.\ Whittaker \cite{whit1935}, but that otherwise had not appeared explicitly in the literature of communication theory. The idea must have been in the air since Nyquist \cite{nyqu1928}; Bennett \cite{benn1941}, in the steady state case; and Gabor \cite{gabo1946} had pointed out that approximately $2TW$ numbers are sufficient to capture a signal of band $W$ Hz that lasts for $T$ seconds.

Further in ``Communication in the presence of noise" \cite{shan1949}, published a year after his seminal ``A mathematical theory of communication" \cite{shan1948}, Shannon establishes a method to represent geometrically any communication system,  and explores the utility of mapping a sequence of samples of a band limited signal into a high dimensional vector space. And it is here where he makes the most interesting of all remarks, on page 13: ``[...] in the case of speech, the ear is insensitive to a certain amount of phase distortion. Messages differing only in the phases of their components [...] sound the same. This may have the effect of {\em reducing the number of essential dimensions in the message space}."

In other words, even if the dimension of the ambient vector space where we embed a representation of a signal is very high, we may come up with an equivalence class for which member points will have similar {\em information content as the original signal}, as far as the end user is concerned; and that equivalence class, in turn, will induce a low dimensional manifold in the vector space where similar messages can be mapped.

These ideas make it natural to frame the theory of compressed sensing \cite{CanTao2005,  CanRomTao2006,dono2006,cand2008,Mack2009} in the context of sampling and information theories. To see this, observe that compressed sensing makes it possible to reconstruct a signal, under certain circumstances, with fewer measurements than the otherwise required number of samples dictated by the Nyquist sampling rate. Moreover, even when the reconstruction is not exact, the error will be small.

In specific, compressed sensing deals with the problem of recovering a signal or message of interest $\mathbf{v} \in \mathbb{R}^n$, which we assume can be represented as $\mathbf{v} = \mathbf{A}\mathbf{x}$ for a matrix $\mathbf{A} \in \mathbb{R}^{n \times N}$, with $n \leq N$, from an incomplete set of linear measurements,
\begin{equation} \label{eq: compressed sensing}
\mathbf{b} = \Phi \mathbf{x}, 
\end{equation}
where $\mathbf{b} \in \mathbb{R}^m$ is the vector of measurements, $\mathbf{x} \in \mathbb{R}^N$ is the object to recover, and $\Phi \in \mathbb{R}^{m \times N}$ is the measurement matrix, with $m < n$ and $\Phi$ is a full rank matrix.\footnote{If $n = N$, we are in the setting of transform coding, where $\mathbf{A}$ represents a unitary transform, for example; and if $n < N$, we can talk of a dictionary or a frame representation of $\mathbf{v}$.} Given a measurement vector $\mathbf{b}$, \cref{eq: compressed sensing} represents an underdetermined system of linear equations, with an infinite number of solutions. However, if $\mathbf{x}$ has at most $k \leq m$ significant components, compared to the rest, we can recover it exactly, or very closely, by solving the constrained problem,
\begin{equation} \label{eq: sparse recovery}
\min_{\mathbf{x} \in \mathbb{R}^N} f(\mathbf{x}) \quad \text{subject to } \| \mathbf{b} - \Phi \mathbf{x}\|_2 = 0,
\end{equation}
where $f(\cdot) = \| \cdot \|_d$, and $d =0 \text{ or } 1$. Here $\| \mathbf{x} \|_0$ counts the number of nonzero entries of $\mathbf{x}$. If $\mathbf{x}^\star$ is a solution to \cref{eq: sparse recovery}, we then synthesize an approximate reconstruction of $\mathbf{v}$ by using $\mathbf{A} \mathbf{x}^\star$.

Note that since $m < n$, we have used fewer measurements than the number of coordinates of $\mathbf{v}$, in effect compressing the sensing, hence the name compressed sensing; possibly beating the Nyquist sampling rate; going from a large dimensional message space, $\mathbb{R}^n$, to a smaller dimensional measurement space, $\mathbb{R}^m$, in a manner that hopefully captures the {\em essence} of the signal of interest. Just like Shannon envisioned.

For all of this to work, we need to make precise the notion of what the ``significant components" of $\mathbf{x}$ are, notion which traditionally has translated into talking of {\em sparsity}. However, we show in \cref{ex: l_0 norm bad at measuring sparcity} that the commonly used notion of sparsity---the number of nonzero entries in a vector---is defective, and we propose instead in \cref{sec: sparsity} a refined notion $s$ of sparsity that extends the traditional meaning of the word as used in the compressed sensing and sparse representation literatures. The definition of $s$ is based on the weak-$\ell^p$ norm, which we define and study in \cref{sec: the weak-lp norm}. The weak-$\ell^p$ norm helps us define, for a given $0 < p \leq 1$, the sparsity function $s_p : \mathbb{R}^n \rightarrow \mathbb{R}$ and the sparsity relation $<_{s_p}$, which induces a strict partial order on $\mathbb{R}^n$. We show that, for a given vector $\mathbf{x} \in \mathbb{R}^n$, $s_p(\mathbf{x}) : \mathbb{R} \rightarrow \mathbb{R}$ is a convex function of $p$, and we use this fact to compute effectively $\inf_{p \in (0,1]} s_p(\mathbf{x})$, which we define as the sparsity $s(\mathbf{x})$ of $\mathbf{x}$. See \cref{sec: defining sparsity}.

In \cref{sec: unitary transforms and sparse representations} we study unitary transformations and the sparsity $s$, which we use to define sparsifying transforms and their properties, formalizing well known energy shifting properties of unitary transforms commonly used in compression, for example. This leads in \cref{sec: error analysis and sparsity} to the study of error analysis and sparsity $s$ when we truncate the signal representation $\mathbf{x}$ of a vector $\mathbf{v}$ under a sparsifying unitary transform $\mathbf{T}$. This is done in terms of the peak signal-to-noise ratio or PSNR, for which we find a lower bound in terms of $s$.

This error analysis and musings on information theoretical matters in \cref{sec: raionale for E} motivate the definition of the sparsity index $\mathcal{E}$ in \cref{sec: sparsity index}, which we use in the context of compressed sensing image reconstruction, by example of the single pixel camera, which is described in detail in \cref{sec: compressed sensing}: In \cref{sec: background} we provide background on the origin of the single pixel camera, in \cref{sec: experiment} we provide a physical realization and mathematical modeling of a single pixel camera, and how to go about obtaining an image from it both in an inefficient way, \cref{sec: naive sensing}, and the compressed sensing way, \cref{sec: compressed sensing details}. In \cref{sec: solving the single pixel camera problem} we show how to solve the single pixel camera compressed sensing problem with either the orthogonal matching pursuit algorithm (OMP), \cref{sec: omp}, or the more efficient and better basis pursuit algorithm (BP), \cref{sec: bp}, for which, in \cref{sec: convex minimization with constraints}, we provide the specific methods that we use to implement it. The characteristics of OMP help us tie in the use of the sparsity index $\mathcal{E}$ with the calculation of a lower bound of the PSNR of the various compressed sensing image reconstructions conducted in \cref{sec: discussion} with BP, given that the solutions obtained with OMP and BP are close. Our results show that we can predict the quality of the reconstruction of images with very good accuracy without knowledge of the original, i.e., we show that we have in the sparsity index $\mathcal{E}$ a reference-free tool to decide when to sample at a higher rate a given region to guarantee a minimum local PSNR reconstruction.

\section{Sparsity} \label{sec: sparsity}
In this section we define the weak-$\ell^p$ norm, go over some of its properties, and use it to redefine the notion of {\em sparsity}, which in common parlance refers to the counting of nonzero entries in a vector. We do this because we show with an example why the commonly used notion of sparsity is not fully satisfactory, and propose instead a new measure of sparsity that utilizes the weak-$\ell^p$ norm, mentioned as a measure of sparsity in \cite{BruDonEla2009}, and used in that capacity in, for example, \cite{CohDeVPetXu1999} and \cite{CanDon2002}. We then derive some properties of this measure of sparsity.

\subsection{The weak-\texorpdfstring{\boldmath$\ell^p$}{lp} norm and its properties} \label{sec: the weak-lp norm}
It is easy to see that given a vector $\mathbf{x} = (x_1, x_2, \ldots, x_n)^\text{T} \in \mathbb{R}^n$, there exists a unique vector $\mathbf{y} = (y_1, y_2, \ldots, y_n)^\text{T} \in \mathbb{R}^n$ satisfying the following two properties:
\begin{enumerate}
\item For all $i \in \{1, 2, \ldots, n\}$, there is a $j \in \{1, 2, \ldots, n\}$ such that $|x_i| = y_j$, and
\item For all $i, j \in \{1, 2, \ldots, n\}$ we have that $y_i \leq y_j \Leftrightarrow i \leq j$.
\end{enumerate}

These two properties naturally define the ordering operator $\Omega : \mathbb{R}^n \rightarrow \mathbb{R}^n$, which assigns to $\mathbf{x}$ its corresponding $\mathbf{y}$. We then write $\mathbf{y} = \Omega(\mathbf{x})$, and say that $\mathbf{y}$ is the ordering of $\mathbf{x}$.

\begin{definition}[Weak-$\ell^p$ norm] \label{def: weaklp}
Let $\mathbf{x} \in \mathbb{R}^n$ and $p > 0$. We define the weak-$\ell^p$ norm of vector $\mathbf{x}$ as the number
\begin{equation*} \label{eq: weaklp}
\| \mathbf{x} \|_{w\ell^p} = \left( \sup_{\epsilon > 0} N(\epsilon,\mathbf{x}) \epsilon^p \right)^\frac{1}{p},
\end{equation*}
where $N(\epsilon,\mathbf{x}) = \#\{j : |x_j| > \epsilon\}$.
\end{definition}
We are interested in the weak-$\ell^p$ norm because for values of $p \in (0,1]$, for a given vector $\mathbf{x}$, the quantity $\|\mathbf{x}\|_{w\ell^p}^p$ can be used as a measure of {\em sparsity} of $\mathbf{x}$. We elaborate on this later on. First, we address how to effectively compute $\|\cdot\|_{w\ell^p}$.

\begin{theorem} \label{theo: computation of weaklp}
Given a vector $\mathbf{x} \in \mathbb{R}^n$, and $p>0$, we have that
\begin{equation} \label{eq: weaklp computation}
\|\mathbf{x}\|_{w\ell^p}^p = \max_i \#\{j : y_j \geq y_i\} y_i^p,
\end{equation}
where $\mathbf{y} = (y_1, y_2, \ldots, y_n)^\text{T} = \Omega(\mathbf{x})$. We define the index $i_\star = i_\star(\mathbf{x}, p)$ as the smallest index where the right hand side of \cref{eq: weaklp computation} reaches its maximum.
\end{theorem}

\begin{proof}
The statement is trivially true for $\mathbf{x} = \mathbf{0}$. Assume then that $\mathbf{x}$ is a nonzero vector with corresponding ordering $\mathbf{y} = \Omega(\mathbf{x})$. First observe that, for a given $\epsilon > 0$, the order in which we count the number of entries $x_j$ in $\mathbf{x}$ that are greater in absolute value than $\epsilon$, does not depend on said order. Therefore, for a given $\epsilon > 0$, we have that $N(\epsilon,\mathbf{x}) = \#\{j : |x_j| > \epsilon\} = \#\{j : y_j > \epsilon\} = N(\epsilon,\mathbf{y})$.

Since $\mathbf{x} \neq \mathbf{0}$, there is an integer $i \in \{1, 2, \ldots, n\}$ such that $y_i > 0$. Let $i_0$ be the smallest of such integers. Consider the following partition $\cal{P}$ of $(0, +\infty) = (0, y_{i_0}) \cup [y_{i_0}, y_{i_0 + 1}) \cup \ldots \cup [y_n, \infty)$. We compute the supremum of $N(\epsilon,\mathbf{y})\epsilon^p$ over each of the intervals defining $\cal{P}$. For $\epsilon \in (0,y_{i_0}) = I_0$, we have that $N(\epsilon, \mathbf{y}) = n - i_0 + 1$, and since raising a number to the $p^\text{th}$ power is a monotonically increasing operation, we clearly have that $\sup_{\epsilon \in I_0} N(\epsilon,\mathbf{y})\epsilon^p = (n - i_0 +1) y_{i_0}^p = \#\{j : y_j \geq y_{i_0}\}y_{i_0}^p$. Similarly, for $\epsilon \in [y_{i_0 + k}, y_{i_0 + k +1}) = I_{k+1}$, we have that $\sup_{\epsilon \in I_{k+1}} N(\epsilon,\mathbf{y})\epsilon^p = \#\{j : y_j \geq y_{i_0+k+1}\}y_{i_0+k+1}^p$. Finally, for $\epsilon \in [y_n, \infty) = I_\infty$, we have that $N(\epsilon,\mathbf{y}) = 0$, and therefore $\sup_{\epsilon \in I_\infty} N(\epsilon,\mathbf{y})\epsilon^p = 0$. The result follows from observing that the supremum of $N(\epsilon,\mathbf{y})\epsilon^p$ over $(0,\infty)$ is the maximum of the supremums of $N(\epsilon,\mathbf{y})\epsilon^p$ over each and all of the intervals $I_0, I_1, \ldots, I_{n-i_0}, I_\infty$.
\end{proof}

We state without proof the following properties of the weak-$\ell^p$ norm, derived from \cref{theo: computation of weaklp}.
\begin{theorem} \label{theo: properties of weak-lp norm}
Let $\mathbf{x} \in \mathbb{R}^n$, $\alpha \in \mathbb{R}$, and $p>0$. Then
\mbox{}
\begin{enumerate}
\item $\| \mathbf{x} \|_{w\ell^p} \geq 0$.
\item $\| \mathbf{x} \|_{w\ell^p} = 0$ if and only if $\mathbf{x} = \mathbf{0}$.
\item $\| \alpha \mathbf{x} \|_{w\ell^p} = |\alpha| \| \mathbf{x} \|_{w\ell^p}$.
\item The weak-$\ell^p$ norm does not satisfy the triangle inequality.
\item $\| \mathbf{x} \|_{w\ell^p} \leq \| \mathbf{x} \|_p$, where $\| \mathbf{x} \|_p = (\sum_i |x_i|^p)^\frac{1}{p}$ is the $\ell^p$-norm.
\end{enumerate}
\end{theorem}

Therefore, the weak-$\ell^p$ norm, is not a true norm, but almost. It is a quasi norm, but for simplicity we will refer to it as a ``norm". We explore and get acquainted with two more properties of the weak-$\ell^p$ norm that will be relevant later on.

From the result of \cref{theo: computation of weaklp}, we observe that the $p^\text{th}$ power of the weak-$\ell^p$ norm of a vector $\mathbf{x} \in \mathbb{R}^n$ corresponds to the largest area of a rectangle of width $\#\{j : y_j \geq y_i\} = n - i + 1$, and height $y_i^p$, where $\mathbf{y} = \Omega(\mathbf{x})$. Recall that in \cref{theo: computation of weaklp} we defined $i_\star$ to be the smallest index for which this maximal area is achieved since we will use $i_\star$ often. For a graphic representation of this concept, see \cref{fig: omega to the pth power,fig: wlp to the p}.

Note that for any value of $p>0$, for a given $i \in \{1, 2, \ldots, n\}$, $\#\{j : y_j \geq y_i\} = n - i + 1$; while $\lim_{p \rightarrow 0^+} y_i^p$ tends to either 1 or 0, depending on whether $y_i \neq 0$ or $y_i = 0$, respectively, i.e., $y_i^p$ tends to the characteristic function $\chi_{\mathbb{R}^*}$ as $p$ goes to zero. Here $\mathbb{R}^* = \mathbb{R} \setminus \{0\}$. It follows that
\begin{equation*}
\lim_{p \rightarrow 0^+} \#\{j : y_j \geq y_i\} y_i^p =
\begin{cases}
n - i + 1, &\text{if } y_i \neq 0, \\
0, &\text{if } y_i = 0.
\end{cases}
\end{equation*}
We conclude from the previous two paragraphs that
\begin{equation} \label{eq: zero norm as a limit}
\lim_{p \rightarrow 0^+} \|\mathbf{x}\|_{w\ell^p}^p = \|\mathbf{x}\|_0 := \#\{ j : x_j \neq 0\}.
\end{equation}
Hence, in the case when $p \rightarrow 0^+$, the weak-$\ell^p$ norm tends to the $\ell_0$-norm, which counts the nonzero entries of a vector, as defined in \cref{eq: zero norm as a limit}. Note that the $\ell_0$-norm is not a norm either, since $\|\alpha \mathbf{x}\|_0 \neq |\alpha| \|\mathbf{x}\|_0$ for $\alpha \in \mathbb{R}^*$ when $\mathbf{x} \neq \mathbf{0}$, but it is commonly called a ``norm" nonetheless.

However, the $\ell_0$-norm is not nuanced at all when we are trying to measure {\em sparsity}, usually defined as the count of the nonzero entries of a vector, in cases where a vector $\mathbf{x} \in \mathbb{R}^n$ has relatively few entries that are considerably larger than the rest in absolute value, a circumstance which we would like to distinguish for reasons that will become clear later on. With this in mind, we propose a new definition and measure of sparsity next.

\begin{figure}[htb]
\centering
\includegraphics[width=0.75\textwidth]{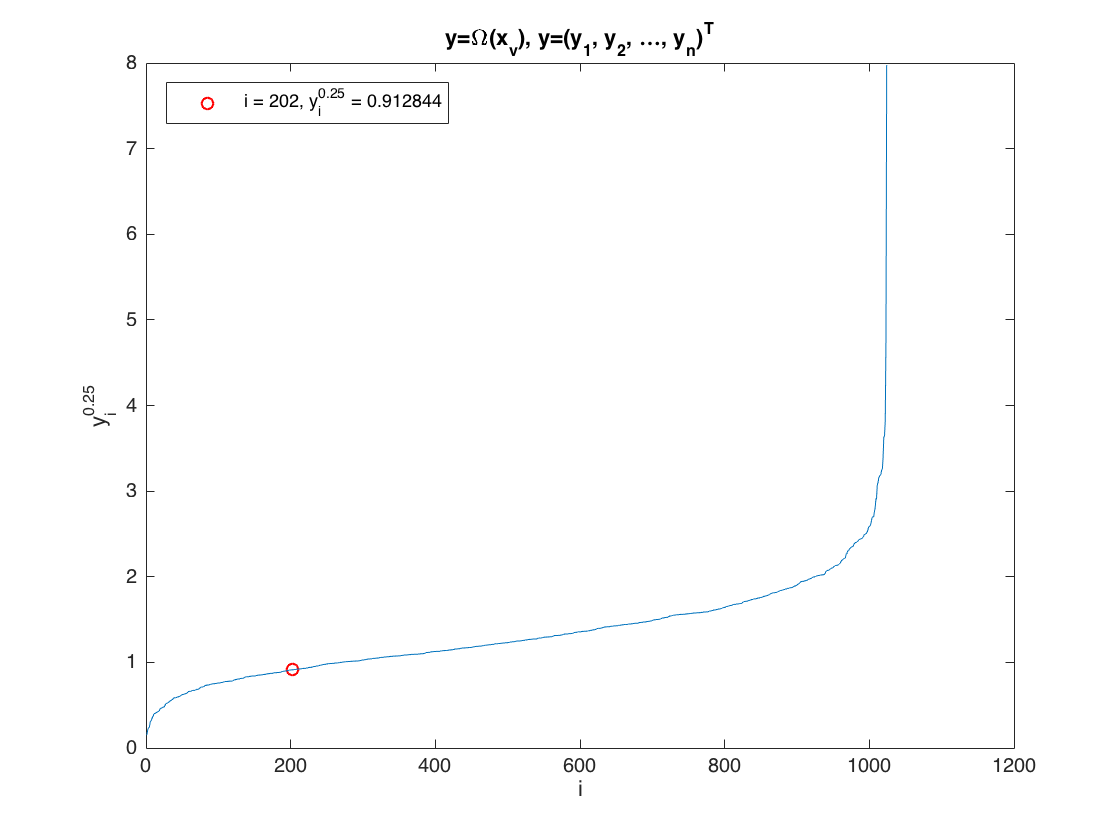}
\caption{Graph of the $p^\text{th}$ power of the entries $y_i$ of $\mathbf{y} = \Omega(\mathbf{x})$ versus $i$. The red circle marks the point $(i_\star,y_{i_\star}^p)$, for $p=\frac{1}{4}$, in this case. See \cref{theo: computation of weaklp} for the definition of $i_\star$.}
\label{fig: omega to the pth power}
\end{figure}

\begin{figure}[htb]
\centering
\includegraphics[width=0.75\textwidth]{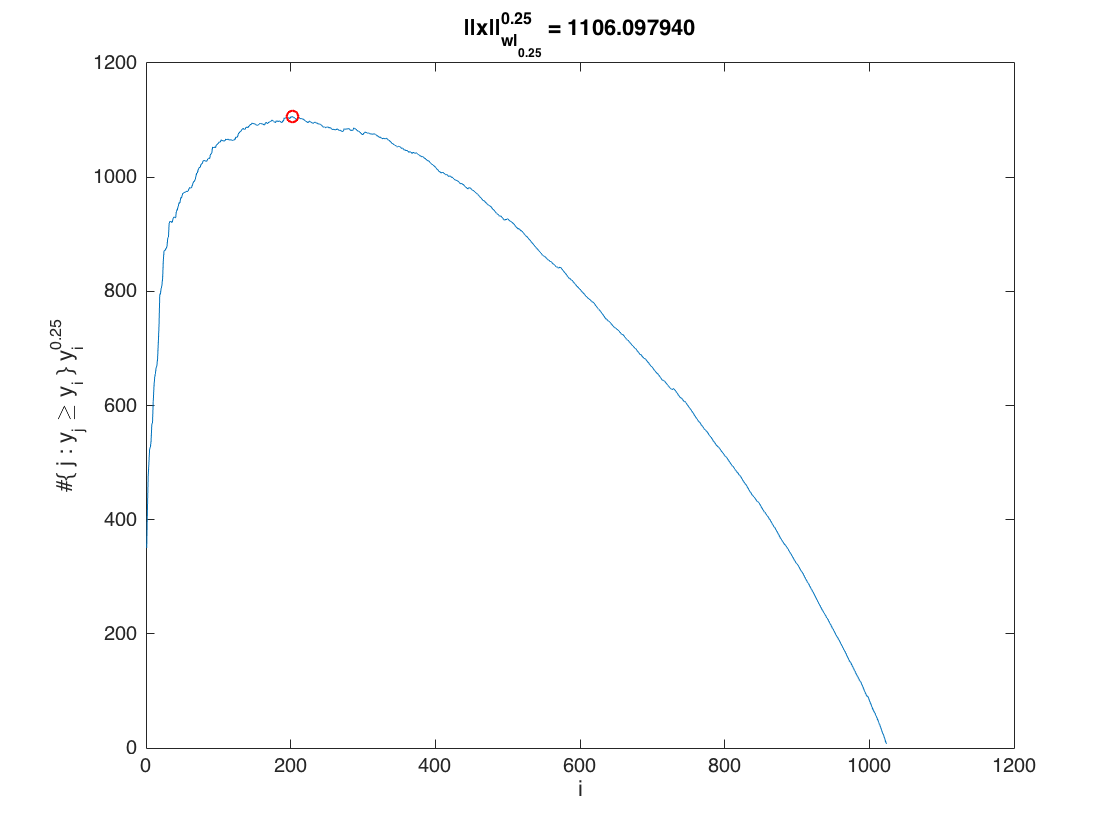}
\caption{Graph of $\#\{j : y_j \geq y_i\} y_i^p$ versus $i$. The red circle marks the point $(i_\star,\|\mathbf{x}\|_{w\ell^p}^p)$, for $p=\frac{1}{4}$. Recall that $\|\mathbf{x}\|_{w\ell^p}^p = \#\{j : y_j \geq y_{i_\star}\} y_{i_\star}^p$. This graph corresponds to the same vector $\mathbf{x}$ used in \cref{fig: omega to the pth power}.}
\label{fig: wlp to the p}
\end{figure}

\subsection{Defining and measuring sparsity} \label{sec: defining sparsity}
In common parlance, as we mentioned in \cref{sec: the weak-lp norm}, we say that a vector $\mathbf{x} = (x_1, x_2, \ldots, x_n)^\text{T} \in \mathbb{R}^n$ is sparse if its $\ell_0$-norm is smaller than $n$. In other words, 
\begin{equation} \label{eq: defining sparsity via ell_0 norm}
\|\mathbf{x}\|_0 = \#\{ j : x_j \neq 0\} < n.
\end{equation}
As argued above, though, this measure of sparsity will not distinguish the following two vectors in $\mathbb{R}^5$ as radically different:
\begin{example} \label{ex: l_0 norm bad at measuring sparcity}
Consider $\mathbf{x}_1=(1, 1, 1, 1, 0)^\text{T}$ and $\mathbf{x}_2=(2, 10^{-16}, 10^{-16}, 10^{-16}, 0)^\text{T}$. From the $\ell_0$-norm point of view, they are both sparse, moreover, their $\ell_0$-norms are equal, yet, most of the entries of $\mathbf{x}_1$ are 1, while most of the entries of $\mathbf{x}_2$ are practically 0.
\end{example}

Clearly, the notion of sparsity defined by \cref{eq: defining sparsity via ell_0 norm} cannot distinguish the very different nature of these two vectors, $\mathbf{x}_1$ and $\mathbf{x}_2$.

Note that in the example above, we deliberately chose both vectors to have approximately equal energy, if we define the energy of a vector $\mathbf{x}$ as $\|\mathbf{x}\|_2^2 = \sum_i x_i^2$. With these observations in hand, we put forth the following definitions.

\begin{definition}[Sparsity $s_p$ and sparsity relation $<_{s_p}$] \label{def: s_p and <_{s_p}}
Let $p \in (0,1]$. Consider the set $X = \mathbb{R}^n$, and define the  binary relation $R \subset X \times X$ as follows:
\begin{equation*}
\text{For all } \mathbf{x}_1, \mathbf{x}_2 \in X, \text{ if } \|\mathbf{x}_1\|_2 = \|\mathbf{x}_2\|_2 \text{ and } \|\mathbf{x}_1\|_{w\ell^p} < \|\mathbf{x}_2\|_{w\ell^p}, \text{ then } (\mathbf{x}_1,\mathbf{x}_2) \in R.
\end{equation*}
We call $R$ the sparsity relation (for $\mathbb{R}^n$ of order $p$), and will write for simplicity $\mathbf{x}_1 <_{s_p} \mathbf{x}_2$ whenever $(\mathbf{x}_1,\mathbf{x}_2) \in R$. If $\mathbf{x}_1 <_{s_p} \mathbf{x}_2$, we say that $\mathbf{x}_1$ is sparser than $\mathbf{x}_2$. We say that $\mathbf{x} \in X$ has sparsity of order $p$ equal to $s_p(\mathbf{x}) := \|\mathbf{x}\|_{w\ell^p}^p$, or simply, that $\mathbf{x}$ has sparsity $s_p(\mathbf{x})$. 
\end{definition}

\begin{theorem} \label{theo: R is a strict partially ordered set}
Let $p \in (0,1]$, then $(\mathbb{R}^n,<_{s_p})$ is a strict partially ordered set.
\end{theorem}
\begin{proof}Let $p \in (0,1]$. For all $\mathbf{x} \in \mathbb{R}^n$, we have that $<_{s_p}$ is trivially irreflexive, i.e., $\mathbf{x} \not<_{s_p} \mathbf{x}$, since $\|\mathbf{x}\|_{w\ell^p} = \|\mathbf{x}\|_{w\ell^p}$, hence $\mathbf{x} \not<_{s_p} \mathbf{x}$. Let $\mathbf{x}_1, \mathbf{x}_2, \mathbf{x}_3 \in \mathbb{R}^n$ and assume that $\mathbf{x}_1 <_{s_p} \mathbf{x}_2$ and $\mathbf{x}_2 <_{s_p} \mathbf{x}_3$. Then, by definition, we must have that $\|\mathbf{x}_1\|_2 = \|\mathbf{x}_2\|_2$ and $\|\mathbf{x}_2\|_2 = \|\mathbf{x}_3\|_2$, as well as $\|\mathbf{x}_1\|_{w\ell^p} < \|\mathbf{x}_2\|_{w\ell^p}$ and $\|\mathbf{x}_2\|_{w\ell^p} < \|\mathbf{x}_3\|_{w\ell^p}$, since both $=$ and $<$ are transitive in $\mathbb{R}$, we have that $\|\mathbf{x}_1\|_2 = \|\mathbf{x}_3\|_2$ and $\|\mathbf{x}_1\|_{w\ell^p} < \|\mathbf{x}_3\|_{w\ell^p}$, and therefore $\mathbf{x}_1 <_{s_p} \mathbf{x}_3$, i.e., $<_{s_p}$ is transitive.
\end{proof}

When we have a partially ordered set, e.g., $(\mathbb{R}^n,<_{s_p})$, we are usually interested in knowing if there are maximal or minimal elements in it with respect to its ordering. Assuming the Axiom of Choice in the form of Zorn's lemma---which states that a partially ordered set in which every chain (i.e., every totally ordered subset), has an upper (lower) bound, necessarily contains at least one maximal (minimal) element---we would then set to find upper (lower) bounds in $\mathbb{R}^n$ for each energy level to conclude that there exist maximal (minimal) elements in $\mathbb{R}^n$ with respect to the partial order $<_{s_p}$. 
We leave the task of establishing the existence of maximal or minimal elements in $(\mathbb{R}^n,<_{s_p})$ for another occasion, since this departs from the focus of our endeavors.

Note that the proof of \cref{theo: R is a strict partially ordered set} does not use anywhere that $p \in (0,1]$ and is, in fact, true for any $p > 0$. However, given the aforementioned observations stemming from \cref{eq: zero norm as a limit} and \cref{eq: defining sparsity via ell_0 norm}, it is clear that measuring sparsity with $s_p$ and comparing the sparsity of two vectors with the sparsity relation $<_{s_p}$, are meaningful and sensible concepts only when $p \in (0,1]$. Therefore, going forward, we will assume that $p \in (0,1]$, unless otherwise noted.


\begin{theorem}[Convexity of $s_p$ as a function of $p$] \label{theo: convexity of s_p}
For all $\mathbf{x} \in \mathbb{R}^n$, the function $f_\mathbf{x} : (0,1] \rightarrow \mathbb{R}$ that maps $p \mapsto s_p(\mathbf{x})$ is a convex function. Moreover, if $\mathbf{y} = (y_1, y_2, \ldots, y_n)^\text{T} = \Omega(\mathbf{x})$ is the ordering of $\mathbf{x}$ and $\mathcal{I}_\star(\mathbf{x}) = \{ i_\star(\mathbf{x},p) : p \in (0,1] \}$  is such that $\{ y_i : i \in \mathcal{I}_\star(\mathbf{x})\} \cap \{0,1\} = \emptyset$, then $f_\mathbf{x}$ is strictly convex. (See \cref{theo: computation of weaklp} for the definition of $i_\star$.)
\end{theorem}

\begin{proof}Let $\mathbf{x} \in \mathbb{R}^n$, $p_1, p_2 \in (0,1]$, $p_1 \neq p_2$, and $t \in [0,1]$. We have that, by \cref{theo: computation of weaklp},
\begin{align} \label{eq: calculation of f_x}
f_\mathbf{x}(t p_1 + (1 - t) p_2) &= s_{t p_1 + (1 - t) p_2}(\mathbf{x}), \nonumber \\
&= \|\mathbf{x}\|_{w\ell_{t p_1 + (1 - t) p_2}}^{t p_1 + (1 - t) p_2}, \nonumber \\
&= \max_i \#\{ j : y_j \geq y_i \} y_i^{t p_1 + (1 - t) p_2}, \nonumber \\
&= \max_i \#\{ j : y_j \geq y_i \} y_i^{t p_1} y_i^{(1 - t) p_2},
\end{align}
where $\mathbf{y} = \Omega(\mathbf{x})$ is the ordering of $\mathbf{x}$. If we prove that, for all $i \in \mathcal{I}_\star(\mathbf{x})$,
\begin{equation} \label{eq: principal inequality}
y_i^{t p_1} y_i^{(1-t) p_2} \leq t y_i^{p_1} + (1-t) y_i^{p_2},
\end{equation}
combining \cref{eq: calculation of f_x} and \cref{eq: principal inequality}, it follows that,
\begin{align} \label{eq: convexity result}
\max_i \#\{ j : y_j \geq y_i \} y_i^{t p_1} y_i^{(1 - t) p_2} &\leq
\max_i \#\{ j : y_j \geq y_i \} \left(t y_i^{p_1} + (1-t) y_i^{p_2}\right), \\
&\leq \max_i \#\{ j : y_j \geq y_i \} t y_i^{p_1} + 
\max_k \#\{ j : y_j \geq y_k \} (1-t) y_k^{p_2}, \nonumber \\
&= t s_{p_1}(\mathbf{x}) + (1-t) s_{p_2}(\mathbf{x}) = t f_\mathbf{x}(p_1) + (1-t) f_\mathbf{x}(p_2), \nonumber
\end{align}
proving that $f_\mathbf{x}$ is convex. Hence, we proceed to prove \cref{eq: principal inequality}. Let $a,b \in \mathbb{R}$, and define the functions $g(t) = e^{ta + (1-t)b}$ and $h(t) = te^a + (1-t)e^b$. We then have that,
\begin{align*}
g(0) = e^b, h(0) = e^b, \\
g(1) = e^a, h(1) = e^a.
\end{align*}
That is, both functions coincide at values $t = 0, 1$. Noting that the graph of $h$ is a line, and observing that $g''(t) = (a-b)^2e^{ta +(1-t)b} \geq 0$, it follows that $g$ is convex, and conclude that $g(t) \leq h(t)$ for all $t \in [0,1]$. Setting $a = \ln(y_i^{p_1})$ and $b = \ln(y_i^{p_2})$, we get that for $t \in [0,1]$,
\begin{align} \label{eq: g convex}
g(t) \leq h(t) \quad \Rightarrow \quad
&e^{t \ln(y_i^{p_1}) + (1-t) \ln{y_i^{p_2}}} \leq t e^{\ln(y_i^{p_1})} + (1-t) e^{\ln(y_i^{p_2})}, \nonumber \\
\Leftrightarrow \quad &y_i^{t p_1} y_i^{(1-t) p_2} \leq t y_i^{p_1} + (1-t) y_i^{p_2},
\end{align}
proving that \cref{eq: principal inequality} holds, as required to complete the first half of the proof. For the second half of the claim, simply note that if $\mathcal{I}_\star(\mathbf{x})$ is such that $\{ y_i : i \in \mathcal{I}_\star(\mathbf{x})\} \cap \{0,1\} = \emptyset$, \cref{eq: principal inequality,eq: convexity result,eq: g convex} become strict inequalities, resulting then in strict convexity for $f_\mathbf{x}$.
\end{proof}

\begin{definition}[Sparsity $s$]
\label{def: sparsity s}We define the sparsity $s$ as the function $s : \mathbb{R}^n \rightarrow \mathbb{R}$ that assigns to every vector $\mathbf{x} \in \mathbb{R}^n$ the number
\begin{equation*}
s(\mathbf{x}) = \inf_{p \in (0,1]} s_p(\mathbf{x}).
\end{equation*} 
\end{definition}

To check that $s$ is well defined, we simply need to prove that for every vector $\mathbf{x} \in \mathbb{R}^n$, $s(\mathbf{x}) \in \mathbb{R}$. Let $p \in (0,1]$. Then, from \cref{theo: properties of weak-lp norm} and \cref{def: s_p and <_{s_p}}, we have that
\begin{equation*}
0 \leq s_p(\mathbf{x}) = \|\mathbf{x}\|_{w\ell^p}^p \leq \|\mathbf{x}\|_p^p < \infty,
\end{equation*}
hence the set $\{ s_p(\mathbf{x}) : p \in (0,1] \}$ is bounded and, therefore, the number $\inf_{p \in (0,1]} s_p(\mathbf{x})$ exists and is unique, which means that $s$ is well defined. Moreover, in light of \cref{theo: convexity of s_p}, computing $s$ can be easily achieved by convex minimization techniques.

It is easy to see that $s$ has the following properties, which we state without proof.

\begin{theorem} \label{theo: properties of s}
Let $\mathbf{x} \in \mathbb{R}^n$, and $\mathbf{y} = (y_1, y_2, \ldots, y_n)^\text{T} = \Omega(\mathbf{x})$ its ordering. Then,
\begin{enumerate}
\item $0 \leq s(\mathbf{x}) \leq \|\mathbf{x}\|_0$.
\item If $\mathbf{y}$ is such that $\mathbf{y} \in Diag(\mathbb{R}^n) = \{ \mathbf{c} \in \mathbb{R}^n : \mathbf{c} = (c ,c ,\ldots, c)^\text{T}, c \in \mathbb{R} \}$, i.e., the ordering of $\mathbf{x}$ is a vector in the diagonal of $\mathbb{R}^n$, then
\begin{equation*}
s(\mathbf{x}) =
\begin{cases}
n, &\text{if } c \geq 1, \\
cn, &\text{if } c < 1,
\end{cases}
\end{equation*}
where $\mathbf{y} = (c, c, \ldots, c)^\text{T}$. Recall that, by definition of $\mathbf{y}$, we must have $c \geq 0$.
\item If $\mathbf{y}$ is such that $y_1 \geq 1$, then $s(\mathbf{x}) = n$.
\end{enumerate}
\end{theorem}

\begin{figure}[htb]
\centering
\includegraphics[width=0.75\textwidth]{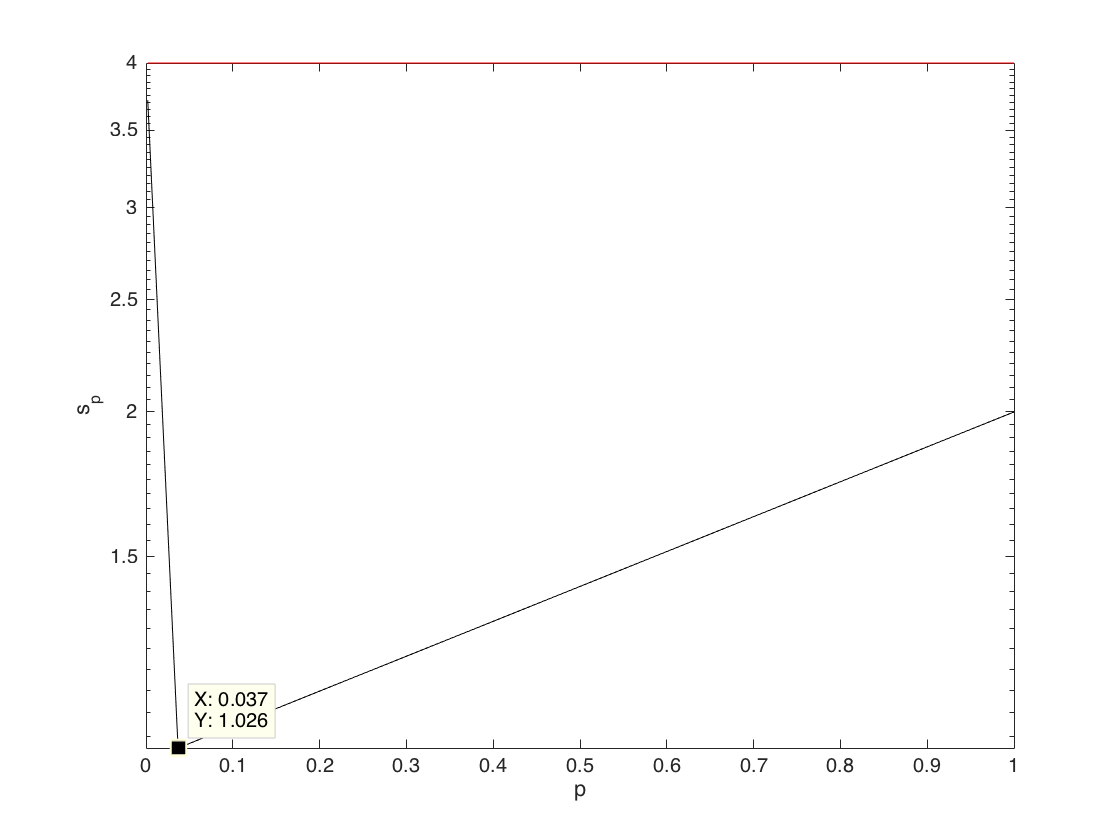}
\caption{Graph of sparsity $s_p(\mathbf{x}_k)$ as a function of $p$ for vectors $\mathbf{x}_1$ and $\mathbf{x}_2$ from \cref{ex: l_0 norm bad at measuring sparcity}. The red line at the top corresponds to $\mathbf{x}_1$, and the black graph, two line segments, corresponds to $\mathbf{x}_2$. Note that we are using a logarithmic scale on the ordinates axis. The point at the bottom shows approximately where $s_p(\mathbf{x}_2)$ reaches its minimum, equal to the sparsity $s(\mathbf{x}_2)$.}
\label{fig: example 1}
\end{figure}

With this new definition of sparsity in hand, we revisit \cref{ex: l_0 norm bad at measuring sparcity} by computing $s(\mathbf{x}_1)$ and $s(\mathbf{x}_2)$, recalling that $\mathbf{x}_1 = (1, 1, 1, 1, 0)^\text{T}$ and $\mathbf{x}_2 = (2, 10^{-16}, 10^{-16}, 10^{-16}, 0)^\text{T}$. This calculation requires from us to compute $s_p(\cdot)$ repeatedly, for which we refer the reader to \cref{theo: computation of weaklp} on how to do it from now on.

We have that $s_p(\mathbf{x}_1) = \max \{ 5 \times 0^p, 4 \times 1^p \} = 4$, and therefore $s(\mathbf{x}_1) = 4$, which, we note, is equal to $\|\mathbf{x}_1\|_0$. Now for $s(\mathbf{x}_2)$, we have that $s_p(\mathbf{x}_2) = \max \{ 5 \times 0^p, 4 \times 10^{-16 p}, 1 \times 2^p \}$. If we draw the graph of $s_p(\mathbf{x}_2)$ as a function of $p$, we see that it is the union of two curves $\Gamma_1$ and $\Gamma_2$, with $\Gamma_1 = \left\{ (p, 4 \times 10^{-16 p}) \in \mathbb{R}^2 : p \in (0, p_0] \right\}$ and $\Gamma_2 = \left\{ (p, 2^p) \in \mathbb{R}^2 : p \in [p_0,1] \right\}$, where $p_0$ is the abscissa such that $4 \times 10^{-16 p_0} = 2^{p_0}$, readily seen as the minimum of $s_p(\mathbf{x}_2)$ over $p \in (0,1]$. It is easy to compute that $p_0 = \frac{\ln(4)}{\ln(2) + 16 \ln(10)} \approx 0.036934$, from which $s(\mathbf{x}_2) \approx 1.025931$, faithfully reflecting the fact that most of the entries in $\mathbf{x}_2$ are practically zero, except for one of them, which is distinctly nonzero. See \cref{fig: example 1}.

\subsection{Unitary transforms and sparse representations} \label{sec: unitary transforms and sparse representations}
In this section we use our new definition of sparsity to explore unitary transforms and sparse representations stemming from them, which we define next.

\begin{definition}[Sparsifying transform and sparse representation] \label{def: sparsifying transform}
Let $\mathbf{T} \in \mathbb{R}^{n \times n}$ be a unitary matrix, and consider the transform $T : \mathbb{R}^n \rightarrow \mathbb{R}^n$ that assigns to every vector $\mathbf{v} \in \mathbb{R}^n$ the vector $T(\mathbf{v}) = \mathbf{T}\mathbf{v} = \mathbf{x}$. We say that $T$ is a sparsifying transform for $V \subset \mathbb{R}^n$ if and only if for every vector $\mathbf{v} \in V$, we have that
\begin{equation*}
\mathbf{x} <_{s_p} \mathbf{v}, \quad \text{for all } p \in (0,1].
\end{equation*}
In this case we say that $\mathbf{T}$ is a sparsifying matrix for $V$, $\mathbf{x}$ is a sparse representation of $\mathbf{v}$ (under $\mathbf{A} = \mathbf{T}^{-1} = \mathbf{T}^\text{T}$), and $\mathbf{v}$ admits ($\mathbf{x}$ as) a sparse representation (under $\mathbf{A} = \mathbf{T}^{-1} = \mathbf{T}^\text{T}$).
\end{definition}

Note that the notion of a sparsifying transform is well defined since it applies to unitary matrices, which preserve energy, i.e., $\|\mathbf{v}\|_2^2 = \|\mathbf{T} \mathbf{v}\|_2^2$ for all $\mathbf{v} \in \mathbb{R}^n$, and therefore $\mathbf{x} = \mathbf{T} \mathbf{v}$ and $\mathbf{v}$ can be compared by the sparsity relation $<_{s_p}$, see \cref{def: s_p and <_{s_p}}.

\begin{theorem}[Sparsity and energy distribution] \label{theo: sparsity and energy distribution}
Let $\mathbf{T} \in \mathbb{R}^{n \times n}$ be a sparsifying matrix for $V \subset \mathbb{R}^n$, and let $\mathbf{v} \in V$ be a vector whose transform is $\mathbf{x} = \mathbf{T} \mathbf{v}$. If $\mathbf{y} = \Omega(\mathbf{x})$ and $\mathbf{w} = \Omega(\mathbf{v})$ are the orderings of $\mathbf{x}$ and $\mathbf{v}$, respectively, then there exists an integer $i_\natural$ such that $y_{i_\natural} < w_{i_\natural}$ and $y_i \geq w_i$ for all $i > i_\natural$. Moreover,
\begin{equation} \label{eq: energy distribution}
\sum_{i \leq i_\natural} y_i^2 \leq \sum_{i \leq i_\natural} w_i^2 \quad \text{and} \quad
\sum_{i > i_\natural} w_i^2 \leq \sum_{i > i_\natural} y_i^2.
\end{equation}
\end{theorem}

\begin{proof}
Let $p \in (0,1]$, and $\mathbf{v} \in V$. Since $\mathbf{T}$ is sparsifying for $V$, we have, by \cref{def: sparsifying transform}, that $s_p(\mathbf{x}) < s_p(\mathbf{v})$, from which, by \cref{theo: computation of weaklp},
\begin{equation*}
(n - i_\star(\mathbf{v},p) + 1) y_{i_\star(\mathbf{v},p)}^p \leq
(n - i_\star(\mathbf{x},p) + 1) y_{i_\star(\mathbf{x},p)}^p <
(n - i_\star(\mathbf{v},p) + 1) w_{i_\star(\mathbf{v},p)}^p,
\end{equation*}
hence $y_{i_\star(\mathbf{v},p)} < w_{i_\star(\mathbf{v},p)}$. Therefore, the set $I = \{i \in \{1, 2, \ldots, n\} : y_i < w_i \} \neq \emptyset$. Let $i_\natural = \max(I)$, be the largest integer in $I$, from which it follows that,
\begin{equation} \label{eq: y less than w at i-natural}
\forall\ i > i_\natural,\ y_i \geq w_i.
\end{equation}
Now, since $\mathbf{T}$ is a unitary matrix, we must have that $\sum_{i=1}^n w_i^2 = \sum_{i=1}^n y_i^2$, from which,
\begin{align*}
&\sum_{i=1}^n (w_i^2 - y_i^2) = 0, \\
\Leftrightarrow \quad &\sum_{i \leq i_\natural} (w_i^2 - y_i^2) +
\sum_{i > i_\natural} (w_i^2 - y_i^2) = 0, \\
\Leftrightarrow \quad & \sum_{i \leq i_\natural} (w_i^2 - y_i^2) =
\sum_{i > i_\natural} (y_i^2 - w_i^2) \geq 0, \ \text{because of \cref{eq: y less than w at i-natural}},
\end{align*}
from which the inequalities in \cref{eq: energy distribution} are easily derived.
\end{proof}
Observe that \cref{theo: sparsity and energy distribution} tells us that the energy in a signal $\mathbf{v} \in \mathbb{R}^n$ gets redistributed into potentially fewer coefficients of its sparse representation $\mathbf{x} = \mathbf{T} \mathbf{v}$, when $\mathbf{T}$ is a sparsifying matrix for $\{\mathbf{v}\}$. We can colloquially say that {\em the energy got squeezed to the right in the ordering of the transform when compared to the ordering of the signal.} See \cref{fig: energy squeezing by a sparsifying transform}, for example.

\begin{figure}[htb]
\centering
\includegraphics[width=0.75\textwidth]{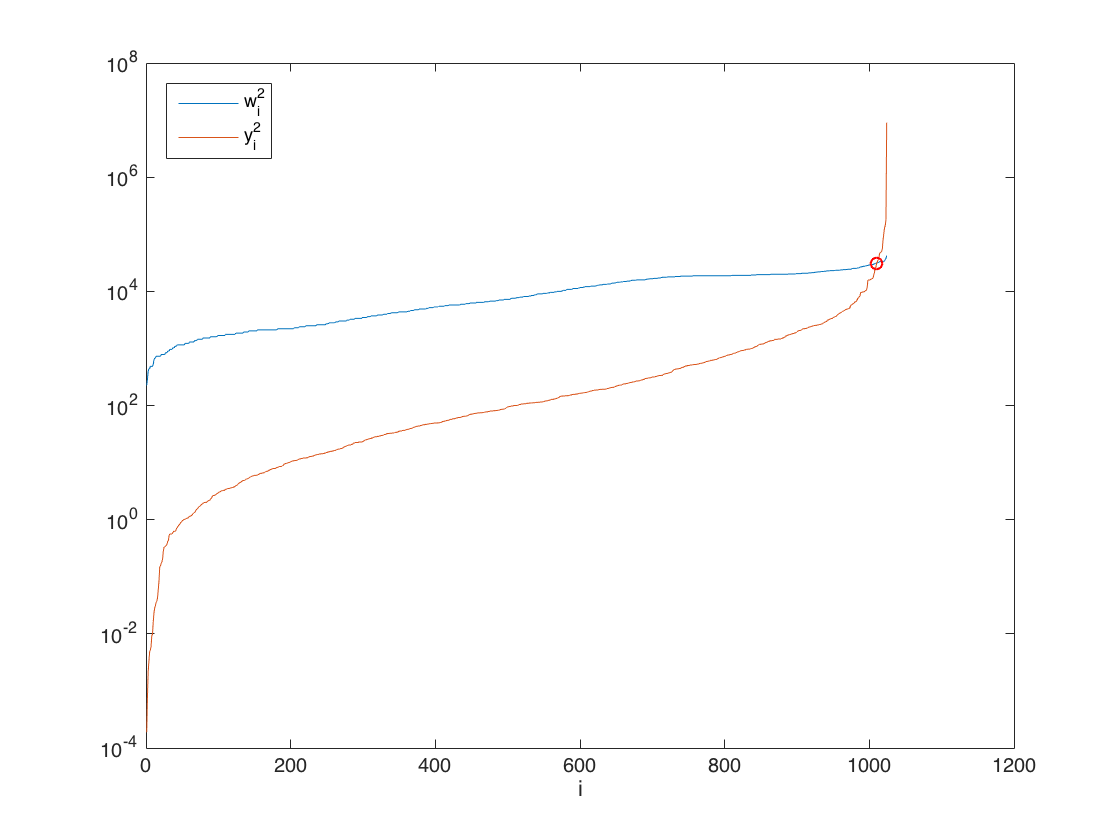}
\caption{Graph of the squared entries of the orderings $\mathbf{w} = \Omega(\mathbf{v})$ and $\mathbf{y} = \Omega(\mathbf{x})$ of a signal $\mathbf{v} \in \mathbb{R}^{1024}$ and its sparse transform $\mathbf{x}$, respectively, under a sparsifying matrix $\mathbf{T}$ for that signal. The red circle denotes the point $(i_\natural,y_{i_\natural}^2)$, with $i_\natural = 1010$ in this example.}
\label{fig: energy squeezing by a sparsifying transform}
\end{figure}

\begin{theorem} \label{theo: s and sparifying transforms}
Let $\mathbf{T} \in \mathbb{R}^{n \times n}$ be a sparsifying matrix for $V \subset \mathbb{R}^n$, and let $\mathbf{v} \in V$ be a vector whose transform is $\mathbf{x} = \mathbf{T} \mathbf{v}$. Then $s(\mathbf{x}) \leq s(\mathbf{v})$. Moreover, if $s(\mathbf{v}) < \|\mathbf{v}\|_0$, then $s(\mathbf{x}) < s(\mathbf{v})$.
\end{theorem}

\begin{proof}
Let $\mathbf{v} \in V$ and $p \in (0,1]$. Since $\mathbf{T}$ is a sparsifying matrix for $V$, $s_p(\mathbf{T} \mathbf{v}) < s_p(\mathbf{v})$. But, by definition, $s(\mathbf{T} \mathbf{v}) = \inf_{p \in (0,1]} s_p(\mathbf{T} \mathbf{v}) \leq s_p(\mathbf{T} \mathbf{v})$, hence $s(\mathbf{T} \mathbf{v})$ is a lower bound for $s_p(\mathbf{v})$. Therefore, $s(\mathbf{T} \mathbf{v}) \leq \inf_{p \in (0,1]} s_p(\mathbf{v}) = s(\mathbf{v})$. Hence, $s(\mathbf{x}) \leq s(\mathbf{v})$, where $\mathbf{x} = \mathbf{T} \mathbf{v}$.

Now assume that $s(\mathbf{v}) < \|\mathbf{v}\|_0$. \Cref{def: s_p and <_{s_p}} and \cref{eq: zero norm as a limit} imply that $\lim_{p \rightarrow 0^+}s_p(\mathbf{v}) =\|\mathbf{v}\|_0$. By definition, this means that for all $\epsilon > 0$, there exists a $\delta > 0$ such that for all $p \in (0,\delta)$, $|s_p(\mathbf{v}) - \|\mathbf{v}\|_0| < \epsilon$. Let $\epsilon_0 = \frac{\|\mathbf{v}\|_0 - s(\mathbf{v})}{2}$, then there exists a $\delta_0 > 0$ such that for all $p \in (0,\delta_0)$ we have that,
\begin{gather*}
|s_p(\mathbf{v}) - \|\mathbf{v}\|_0 | < \epsilon_0 = \frac{\|\mathbf{v}\|_0 - s(\mathbf{v})}{2}, \\
\Rightarrow \quad -s_p(\mathbf{v}) + \|\mathbf{v}\|_0 <
	\frac{\|\mathbf{v}\|_0 - s(\mathbf{v})}{2}, \\
\Leftrightarrow \quad -s_p(\mathbf{v}) < - \frac{\|\mathbf{v}\|_0 + s(\mathbf{v})}{2}, \\
\Leftrightarrow \quad s_p(\mathbf{v}) > \frac{\|\mathbf{v}\|_0 + s(\mathbf{v})}{2} >
	 s(\mathbf{v}), \\
\Rightarrow \quad \inf_{p \in (0,\delta_0)} s_p(\mathbf{v}) > s(\mathbf{v}) =
	\inf_{p \in (0,1]} s_p(\mathbf{v}) = \inf_{p \in [\delta_0,1]} s_p(\mathbf{v}).
\end{gather*}
Since $s_p(\mathbf{v})$ is a continuous function of $p$, and $[\delta_0, 1]$ is compact, there exists a $p_\star \in [\delta_0, 1]$ such that $\inf_{p \in [\delta_0, 1]} s_p(\mathbf{v}) = s_{p_\star}(\mathbf{v})$, therefore $s(\mathbf{v}) = s_{p_\star}(\mathbf{v})$. Since $\mathbf{T}$ is a sparsifying matrix for $V$, we have that $s(\mathbf{x}) \leq s_{p_\star}(\mathbf{x}) < s_{p_\star}(\mathbf{v}) = s(\mathbf{v})$. Hence, $s(\mathbf{x}) < s(\mathbf{v})$.
\end{proof}

In a similar but opposite observation to what happens to the energy in view of \cref{theo: sparsity and energy distribution}, here, the sparsity of the transform $\mathbf{x}$ of a signal $\mathbf{v}$ under a sparsifying matrix $\mathbf{T}$ gets shifted to the {\em left} of the sparsity value of said signal.

In the proof of \cref{theo: s and sparifying transforms}, given a vector $\mathbf{x} \in \mathbb{R}^n$, we used the notation $p_\star$ to talk about a value of $p$ for which $s_p(\mathbf{x})$ reaches its minimum as a function of $p$ on $(0,1]$, resulting in $s(\mathbf{x}) = s_{p_\star}(\mathbf{x})$. If $\mathbf{x}$ is such that $\{ y_i : (y_1, y_2, \ldots, y_n)^\text{T} = \Omega(\mathbf{x}), i \in \mathcal{I}_\star(\mathbf{x})\} \cap \{0,1\} = \emptyset$, by \cref{theo: convexity of s_p}, $s_p(\mathbf{x})$ is strictly convex as a function of $p$, making $p_\star$ unique in this case. When $s(\mathbf{x}) = \|\mathbf{x}\|_0$, from \cref{eq: zero norm as a limit}, we can set $p_\star = 0$, and think of $p_\star$ as the unique value of $p \in [0,1]$ for which $s(\mathbf{x}) = s_{p_\star}(\mathbf{x})$. These results and observations can be summarized in the following theorem.

\begin{theorem} \label{theo: unique p_star}
Given a vector $\mathbf{x} \in \mathbb{R}^n$, if $\{ y_i : (y_1, y_2, \ldots, y_n)^\text{T} = \Omega(\mathbf{x}), i \in \mathcal{I}_\star(\mathbf{x})\} \cap \{0,1\} = \emptyset$, then there is a unique $p_\star \in [0,1]$ such that $s(\mathbf{x}) = s_{p_\star}(\mathbf{x})$. Recall that, $\mathcal{I}_\star(\mathbf{x}) = \{ i_\star(\mathbf{x},p) : p \in (0,1] \}$, and that $i_\star = i_\star(\mathbf{x},p)$ is the smallest integer such that $s_p(\mathbf{x}) = \|\mathbf{x}\|_{w\ell^p}^p = \#\{j : y_j \geq y_{i_\star}\} y_{i_\star}^p$.
\end{theorem}

A couple of remarks are in order. The rather technical condition in \cref{theo: unique p_star}, for a given vector $\mathbf{x} \in \mathbb{R}^n$, that
\begin{equation} \label{eq: unicity of p condition}
\{ y_i : (y_1, y_2, \ldots, y_n)^\text{T} =
\Omega(\mathbf{x}), i \in \mathcal{I}_\star(\mathbf{x})\} \cap \{0,1\} = \emptyset,
\end{equation}
is necessary for there to be a unique value $p_\star$ such that $s(\mathbf{x}) = s_{p_\star}(\mathbf{x})$, is not uncommon when $\mathbf{x}$ is a random or semi-structured vector. We don't have a proof of this statement, but it is our empirical observation that all vectors $\mathbf{x}$ that are the transform of some {\em real life} vector $\mathbf{v} \in \mathbb{R}^n$, such as a {\em natural image}, under a unitary matrix $\mathbf{T} \in \mathbb{R}^{n \times n}$, satisfy the condition summarized in \cref{eq: unicity of p condition}. Moreover, even if the set of minimizers $\mathcal{P}_\star(\mathbf{x})$ of $f_\mathbf{x}(\cdot) = s_{(\cdot)}(\mathbf{x})$ is not a singleton, based on numerical experiments, we would venture the educated guess (not a conjecture-level claim, really) that $\{ i_\star(\mathbf{x},p) : p \in \mathcal{P}_\star(\mathbf{x})\}$ is a singleton.

Therefore, it is not too much to sacrifice, for the work ahead of us, to assume that all vectors $\mathbf{x}$ that we will encounter satisfy \cref{eq: unicity of p condition} or, even less restrictively, that $\{ i_\star(\mathbf{x},p) : p \in \mathcal{P}_\star(\mathbf{x})\}$ is a singleton. With these assumptions in mind then, we can talk of {\em the} value of $p_\star$ for which $s(\mathbf{x}) = s_{p_\star}(\mathbf{x})$, or {\em the} value of $i_\star$ for which $s(\mathbf{x}) = \#\{j : y_j \geq y_{i_\star}\} y_{i_\star}^p$, for any $p \in \mathcal{P}_\star(\mathbf{x})$, unambiguously.

\subsection{Error analysis and sparsity} \label{sec: error analysis and sparsity}
Consider a signal $\mathbf{v} \in \mathbb{R}^n$ and its transform $\mathbf{x} \in \mathbb{R}^n$ under $T : \mathbb{R}^n \rightarrow \mathbb{R}^n$, which maps $\mathbf{v} \mapsto \mathbf{x} = \mathbf{T} \mathbf{v}$, where $\mathbf{T} \in \mathbb{R}^{n \times n}$ is a unitary matrix. Assume that $\mathbf{x}$ has sparsity $s(\mathbf{x})$ with corresponding $p_\star \in [0,1]$ and integer $i_\star$, such that $s(\mathbf{x}) = s_{p_\star}(\mathbf{x}) = \#\{j : y_j \geq y_{i_\star}\} y_{i_\star}^{p_\star}$. Here, as usual, $\mathbf{y} = (y_1, y_2, \ldots, y_n)^\text{T} = \Omega(\mathbf{x})$ is the ordering of $\mathbf{x}$.

Imagine now that we choose an integer $i_0 < i_\star$ and form a vector $\tilde{\mathbf{x}}$ equal to $\mathbf{x}$ except that we zero out its $i_0$ smallest entries in absolute value. This implies that if $\tilde{\mathbf{y}} = \Omega(\tilde{\mathbf{x}})$, then $\tilde{y}_i = 0$ for $i \in \{1, 2, \ldots, i_0\}$, and $\tilde{y}_i = y_i$ for $i \in \{i_0+1, i_0 + 2, \ldots, n\}$. Moreover, since $\mathbf{T}$ is unitary, if we set $\tilde{\mathbf{v}} = \mathbf{T}^{-1} \tilde{\mathbf{x}}$, we then have that
\begin{align} \label{eq: square error upper bound}
&\|\mathbf{v} - \tilde{\mathbf{v}}\|_2^2 = \| \mathbf{x} - \tilde{\mathbf{x}}\|_2^2 
= \sum_{i = 1}^{i_0} y_i^2 \leq \sum_{i=1}^{i_0} y_{i_0}^2, \nonumber \\
\Rightarrow \quad & \left(\sqrt{\sum_{i = 1}^{i_0} y_i^2}\right)^{p_\star} 
\leq \sqrt{i_0}^{p_\star} y_{i_0}^{p_\star} < \sqrt{i_0}^{p_\star} y_{i_\star}^{p_\star}.
\end{align}
Since,
\begin{align} \label{eq: value of y to the p}
\|\mathbf{x}\|_{w\ell^{p_\star}}^{p_\star} &= (n - i_\star + 1) y_{i_\star}^{p_\star}, \nonumber \\
\Leftrightarrow \quad y_{i_\star}^{p_\star} 
&= \frac{\|\mathbf{x}\|_{w\ell^{p_\star}}^{p_\star}}{n - i_\star + 1},
\end{align}
we combine \cref{eq: square error upper bound} and \cref{eq: value of y to the p} to obtain,
\begin{align} \label{eq: mse bound by weaklp norm}
&\left(\sqrt{\sum_{i = 1}^{i_0} y_i^2}\right)^{p_\star} <
\sqrt{i_0}^{p_\star} \frac{\|\mathbf{x}\|_{w\ell^{p_\star}}^{p_\star}}{n - i_\star + 1}, \nonumber \\
\Leftrightarrow \quad & \sqrt{\sum_{i = 1}^{i_0} y_i^2} <
\frac{\sqrt{i_0}}{(n-i_\star+1)^{1/p_\star}} \|\mathbf{x}\|_{w\ell^{p_\star}}, \nonumber \\
\Leftrightarrow \quad & \frac{1}{n} \sum_{i = 1}^{i_0} y_i^2 <
\frac{i_0}{n(n-i_\star+1)^{2/p_\star}} \|\mathbf{x}\|_{w\ell^{p_\star}}^2.
\end{align}
Notice that the left hand side of \cref{eq: mse bound by weaklp norm} is the {\em mean squared error} between $\mathbf{x}$ and $\tilde{\mathbf{x}}$, noted $MSE(\mathbf{x},\tilde{\mathbf{x}})$, which is equal to $MSE(\mathbf{v},\tilde{\mathbf{v}})$ given that $\mathbf{T}$ is a unitary matrix. In terms of the {\em peak signal-to-noise ratio} between $\mathbf{v}$ and $\tilde{\mathbf{v}}$, or $PSNR(\mathbf{v},\tilde{\mathbf{v}})$, we have that, by definition \cite{TauMar2002},
\begin{align} \label{eq: psnr lower bound}
PSNR(\mathbf{v},\tilde{\mathbf{v}})
&= 10 \log_{10}\left(\frac{MAX^2}{MSE(\mathbf{v},\tilde{\mathbf{v}})}\right), \nonumber \\
\Rightarrow \quad PSNR(\mathbf{v},\tilde{\mathbf{v}})
&> 10 \log_{10}\left(\frac{n(n-i_\star+1)^{2/p_\star} MAX^2}
{i_0\, \|\mathbf{x}\|_{w\ell^{p_\star}}^2}\right),
\end{align}
which follows from \cref{eq: mse bound by weaklp norm}. Here, $MAX$ is the maximum absolute value that the entries of $\mathbf{v}$ can reach. \Cref{eq: psnr lower bound} gives us a lower bound for the peak signal-to-noise ratio between $\mathbf{v}$ and $\tilde{\mathbf{v}}$ in terms of the sparsity of its transform $\mathbf{x}$, $s(\mathbf{x}) = \|\mathbf{x}\|_{w\ell^{p_\star}}^{p_\star}$---with its corresponding values of $i_\star$ and $p_\star$---and the number $i_0$ of the smallest entries in absolute value of $\mathbf{x}$ that we decided to set to zero. We rewrite \cref{eq: psnr lower bound} to make this observation explicit in mathematical terms,
\begin{equation} \label{eq: psnr lower bound in terms of sparsity}
PSNR(\mathbf{v},\tilde{\mathbf{v}}) >
10 \log_{10}\left(\frac{n(n-i_\star+1)^{2/p_\star} MAX^2}
{i_0\, s(\mathbf{x})^{2/p_\star}}\right).
\end{equation}
The right hand side of \cref{eq: psnr lower bound in terms of sparsity} can be written as follows,
\begin{align} \label{eq: psnr lower bound detailed calculation}
10 \log_{10}\left(\frac{n(n-i_\star+1)^{2/p_\star} MAX^2}{i_0\, s(\mathbf{x})^{2/p_\star}}\right) =
10 \log_{10}\left(\frac{n^{2/p_\star}n^{1-2/p_\star}(n-i_\star+1)^{2/p_\star} MAX^2}{i_0\, s(\mathbf{x})^{2/p_\star}}\right), \notag\\
= 10 \log_{10}\left( \left(\frac{n}{s(\mathbf{x})}\right)^{2/p_\star}\frac{n^{1-2/p_\star}(n-i_\star+1)^{2/p_\star} MAX^2}{i_0}\right), \notag\\
= \underbrace{-\frac{20}{p_\star}\log_{10}\left(\frac{s(\mathbf{x})}{n}\right)}_{\sigma_1} + \underbrace{10\frac{p_\star - 2}{p_\star}\log_{10}(n)}_{\sigma_2} + \underbrace{\frac{20}{p_\star}\log_{10}(n - i_\star +1)}_{\sigma_3} \\
+ \underbrace{20\log_{10}(MAX)}_{\sigma_4} + \underbrace{-10\log_{10}(i_0)}_{\sigma_5}. \notag
\end{align}

From \cref{eq: psnr lower bound detailed calculation}, we can study the dependence of the lower bound of $PSNR(\mathbf{v},\tilde{\mathbf{v}})$ in \cref{eq: psnr lower bound in terms of sparsity} as a function of various parameters and quantities that we address in detail next.

From \cref{theo: properties of s}, we have that $0 \leq s(\mathbf{x})/n \leq 1$ and therefore $-\infty \leq \log_{10}(s(\mathbf{x})/n) \leq 0$. Hence, since $0 \leq p_\star \leq 1$, and assuming that $p_\star \neq 0$, the first summand $\sigma_1$ in \cref{eq: psnr lower bound detailed calculation} satisfies $0 \leq \sigma_1 \leq \infty$, $\lim_{s(\mathbf{x}) \rightarrow 0} \sigma_1 = \infty$, and $\sigma_1|_{s(\mathbf{x}) = n} = 0$. In other words, $\sigma_1$ can only but increase the lower bound of $PSNR(\mathbf{v},\tilde{\mathbf{v}})$; for a given value of $p_\star$, the smaller the value of the sparsity $s(\mathbf{x})$, the better; and its contribution is null when $s(\mathbf{x})=n$.

For the second summand $\sigma_2$, since $0 \leq p_\star \leq 1$, its contribution to the lower bound of $PSNR(\mathbf{v},\tilde{\mathbf{v}})$ can only be but a negative number. Observe that for a given $p_\star$, the contribution of this summand gets worse the larger the dimension $n$ is.

In the case of $\sigma_3$, for a given value of $p_\star$ and a dimension $n$, the smaller $i_\star$ the better.

The summand $\sigma_4$ is a constant that is dependent on the dynamic range of the signal, represented here by $MAX$. The bigger the dynamic range $MAX$, the larger the lower bound for $PSNR(\mathbf{v},\tilde{\mathbf{v}})$ is.

Finally, unsurprisingly, the last summand $\sigma_5$ is at best $\infty$ for $i_0 = 0$, presumably overruling all other terms making the lower bound for $PSNR(\mathbf{v},\tilde{\mathbf{v}})$ infinite, since we are not removing any terms of the representation $\mathbf{x}$ of $\mathbf{v}$ under $\mathbf{T}$, resulting in $\mathbf{v} = \tilde{\mathbf{v}}$. However, its contribution is null when $i_0=1$, and progressively worse as $i_0$ increases, with a worse case scenario of $\sigma_5 = -10 \log_{10}(n)$.

\subsection{The sparsity index \texorpdfstring{\boldmath$\mathcal{E}$}{E}} \label{sec: sparsity index}
Motivated by the results in \cref{sec: error analysis and sparsity} that show the link that exists between the error incurred when we truncate the smallest coefficients in absolute value of a unitary representation of a signal and the sparsity of its transform, we define what we call the sparsity index $\mathcal{E}$. In \cref{sec: raionale for E} we provide some background information and ideas that form the basis for the rationale of the choice of the discrete cosine transform (DCT) \cite{AhmNatRao1974,RaoYip1990} as the unitary transform at the core of its definition. The definition of $\mathcal{E}$ that we give below is geared for use on 2D-data because we are interested in image processing. However, it can be modified for 1D-data, for use in time-series, for example, in a natural way by simply taking the 1D DCT instead of the 2D DCT.
\begin{definition}[Sparsity index $\mathcal{E}$]
\label{def: sparsity index}Let $k,l \in \mathbb{N}^*$, $n=k \times l$, and $\mathbf{V}$ be a matrix in $\mathbb{R}^{k \times l}$. Let $\mathbf{X} = dct2(\mathbf{V})$ be the 2D DCT of $\mathbf{V}$, and let $\mathbf{x} \in \mathbb{R}^n$ be the vector that results from stacking the columns of $\mathbf{X}$. We then define the sparsity index $\mathcal{E}$ of $\mathbf{V}$ as
\begin{equation} \label{eq: sparsity index}
\mathcal{E}(\mathbf{V}) = \frac{s(\mathbf{x})}{n}.
\end{equation}
\end{definition}
Observe that for any $\mathbf{V} \in \mathbb{R}^{k \times l}$, we have that, by \cref{theo: properties of s}, $0 \leq \mathcal{E}(\mathbf{V}) \leq 1$.

We will show how to use the sparsity index $\mathcal{E}$ to predict the error in the reconstruction of a signal in the setting of compressed sensing in \cref{sec: discussion}.

\section{A compressed sensing example: The single pixel camera} \label{sec: compressed sensing}
The explosion of activity in the field of sparse and redundant representations, spanning two decades by 2010, gave us algorithms, and theoretical results guaranteeing their performance, to approximate the sparsest solutions\footnote{Sparsity here refers to the usual $\| \cdot \|_0$ notion of sparsity, not the one embodied in our function $s$, see \cref{def: sparsity s}.} of linear systems of equations \cite{elad2010}. One of the applications in this field is compressed sensing, which we described briefly in \cref{sec: introduction}. We use this signal processing technique in the context of images next.

\subsection{Background} \label{sec: background}
A modern consumer camera typically contains a single charge-coupled device (CCD) or complementary metal-oxide semiconductor (CMOS) sensor that captures light on a regular grid of picture elements, called pixels. The intensity of light falling on each individual pixel is translated into a numerical value, and theses quantities are in turn processed to render an image.

A Bayer color filter array (CFA) \cite{Bay1976} is typically used on the surface of the sensor to obtain red, green, and blue color light sample values at specific pixel locations. The processing of scenes captured with a Bayer CFA requires the extra processing step of demosaicing to produce full color images, as opposed to the simpler case of processing grayscale images, cf., \cite{PlaVen2000}.

In today's age of megapixel cameras, we can cheaply manufacture a sensor with millions of pixels that is sensitive to the visible light spectrum. Problems arise when we desire similar resolutions for light spectra where the sensors are much more expensive, e.g., infrared or ultra-violet sensors.

Richard G.\ Baraniuk et al.\ have constructed a single-pixel camera using a digital micromirror device (DMD) and compressed sensing techniques to produce grayscale images, see \cite{DuDaTaLaSuKeBa2008}. In order to go beyond DMD technology, new concepts and designs for the construction of a single-pixel camera that utilizes a liquid crystal display (LCD) are underway. The idea is to simulate a pixel grid with an LCD filter and ``sum up'' the resulting light that comes through it with only a single-pixel light sensor. One hopes that if this is performed correctly, we can obtain a resolution equal to that of the LCD display.

\subsection{The experiment and its mathematical modeling} \label{sec: experiment}

\subsubsection{Naive sensing} \label{sec: naive sensing}

The experimental design that we use is outlined in \cref{fig: experimental setup}. Ambient light reflects off the target and passes through the front lens. This lens focuses the light into a beam which is directed at an LCD. The LCD is a grid of squares, say 1024 by 768, which are equivalent to pixels in a sensor. Each square can be switched on or off. This either allows light to pass through, or not, said square, respectively. The total admitted light is captured by the back lens, which then concentrates the light into a single-pixel CCD or CMOS sensor. This sensor counts the incoming photons and gives out a corresponding output voltage, which can be measured.

\begin{figure}[htp]
    \centering
    \includegraphics[width=0.62\textwidth]{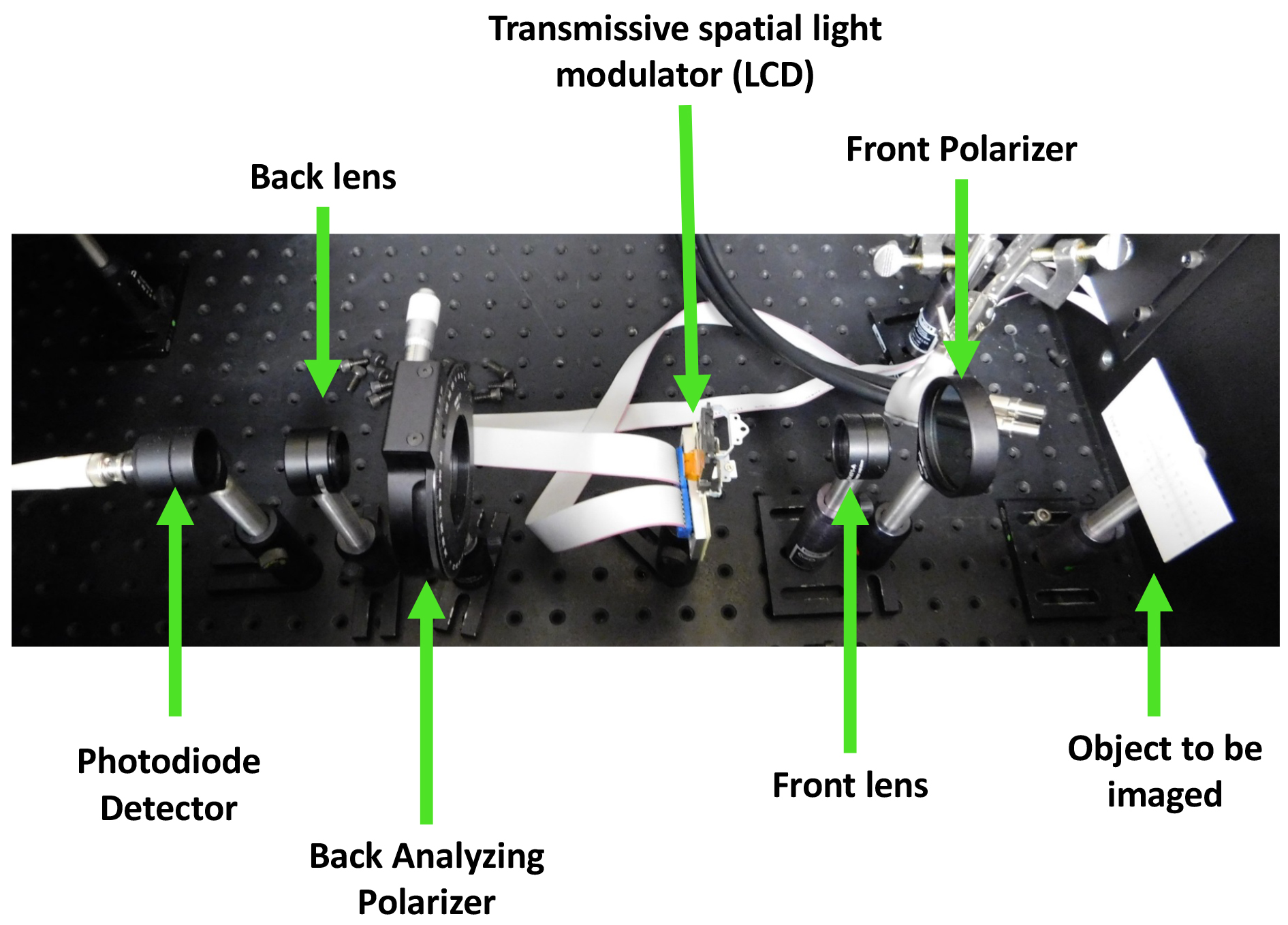}
    \caption{Physical experimental design. Photo modified from an original provided courtesy of Dr. David Bowen, Laboratory for Physical Sciences.}
    \label{fig: experimental setup}
\end{figure}

To model this mathematically, first we can imagine our image as a function $f: \Omega \to \mathbb{R}$,
where $\Omega \subseteq \mathbb{R}^2$ can be thought of as the plane containing the image to be captured and the output $f(\mathbf{q})$ can be thought of as the intensity of light per surface unit, i.e., intensity density, at point $\mathbf{q} \in \Omega$. We can then discretize the image by splitting it into a rectangular grid of size $m \times n$, where each rectangle corresponds to a single image pixel. We determine the value of light intensity at pixel $(i,j)$ by computing the integral
\begin{equation*}
    v_{i,j} = \int_{Q_{i,j}} f(\mathbf{q}) \, d\mathbf{q},
\end{equation*}
where $Q_{i,j} \subset \Omega$ is the square on the image that corresponds to pixel $(i,j)$. In other words, $v_{i,j}$ is the light intensity of the image at square $Q_{i,j}$, which is the image sample value associated to pixel $(i,j)$. The LCD can be modeled as a vector $\mathbf{p} = (p_0, p_1, \ldots, p_{mn-1})^\text{T}$ of length $mn$, where each entry $p_{k}$ is either 0 or 1. An entry $p_k = 1$ corresponds to letting the light from $Q_{i,j}$ pass through, and $p_k = 0$ corresponds to blocking it. Note that we have implicitly defined a bijection $(i,j) \xleftrightarrow{vec} k$, where $vec$ maps the coordinates $(i,j)$ to their corresponding position $k$ in $\mathbf{p}$. The sensor at the end captures the total light intensity of all the image squares that were not blocked by mask $\mathbf{p}$. Its value is given by the sum,
\begin{equation} \label{eq: light measurement given mask p}
y_\mathbf{p} = \sum_{i = 0}^{m-1} \sum_{j = 0}^ {n-1} v_{i,j} p_{vec(i,j)}.
\end{equation}

It is easy to see from \cref{eq: light measurement given mask p} that if we set $p_{vec(0,0)} = 1$, and $p_{vec(i,j)} = 0$ for $i,j \neq 0$, then $y_{\mathbf{p}} = v_{0,0}$. Making a similar arrangement for all possible coordinate pairs $(i,j)$ we can construct $mn$ vectors $\mathbf{p}$ that  allow us to recover all the $mn$ values $\{v_{i,j}\}$ necessary to recover the discretized image implied by $f$ when using an $m \times n$ LCD. However, this is not very efficient. We can do much better than this.

\subsubsection{Compressed sensing} \label{sec: compressed sensing details}
With a slight adjustment to the formulation given at the end of \cref{sec: naive sensing}, our experimental setup fits into the framework of compressed sensing, which we described in \cref{sec: introduction}. The overall goal is to take $K \ll mn$ measurements and still recover the discretized image described by the set of values $\{v_{i,j}\}$. We organize $\{v_{i,j}\}$ into a vector $\mathbf{v} \in \mathbb{R}^M$, where $M = mn$, utilizing the bijection $vec$, introduced in \cref{sec: naive sensing}, by setting $\mathbf{v} = (v_0, v_1, \ldots, v_{mn-1})^\text{T}$, where $v_{k} = v_{i,j}$ if and only if $k = vec(i,j)$. Then we can write \cref{eq: light measurement given mask p} as the inner product
\begin{equation*}
y_\mathbf{p} = \langle \mathbf{v},\mathbf{p} \rangle = \mathbf{p}^\text{T} \mathbf{v}.
\end{equation*}
We repeat this process to collect $K$ samples $\{y_{\mathbf{p}_l}\}$, with $K$ distinct $\{\mathbf{p}_l\}$ sampling masks, each corresponding to a different LCD configuration. We write in condensed form all $K$ measurements in matrix notation
\begin{equation*}
\mathbf{y} = \mathbf{P}^\text{T} \mathbf{v},
\end{equation*}
where $\mathbf{P} = (\mathbf{p}_0\ \mathbf{p}_1\ \ldots\ \mathbf{p}_{K-1})$ is the measurement matrix formed with the $K$ column vectors $\{\mathbf{p}_l\}$ corresponding each to a different LCD configuration.

Finally, if there is a basis or frame $\mathbf{A} \in \mathbb{R}^{M \times N}$, with $M \leq N$, where image $\mathbf{v}$ has a sparse representation, then we can formulate the problem of reconstructing $\mathbf{v}$ as a compressed sensing problem:

Given $\mathbf{y} \in \mathbb{R}^K$ and $\mathbf{P} \in \mathbb{R}^{M \times K}$, find $\mathbf{x}^\star \in \mathbb{R}^N$ solving
\begin{equation} \label{eq: compressed sensing for single pixel camera}
\min_{\mathbf{x} \in \mathbb{R}^N} f(\mathbf{x}), \quad \text{subject to } 
\left\| \mathbf{y} - \mathbf{P}^\text{T} \mathbf{A} \mathbf{x} \right\|_2 = 0,
\end{equation}
with $f(\mathbf{x}) = \|\mathbf{x}\|_0$, or $f(\mathbf{x})= \|\mathbf{W}\mathbf{x}\|_1$, where $\mathbf{W}$ is the diagonal matrix with $i$th diagonal entry $w_i = \|(\mathbf{P}^\text{T}\mathbf{A})_i\|_2$, the $\ell^2$-norm of the $i$th column of $\mathbf{P}^\text{T}\mathbf{A}$.

If image $\mathbf{v}$ is $\epsilon >0$ close to a sparse representation in the basis or frame $\mathbf{A}$ with at most $s$ nonzero elements, then the theory of compressed sensing can guarantee that we can find a solution $\mathbf{x}^\star$ to \cref{eq: compressed sensing for single pixel camera} provided $K = O(s \log(N/s))$. Algorithms for finding such minimizing $\mathbf{x}^\star$ include the Orthogonal Matching Pursuit (OMP) algorithm \cite{PatRezKri1993,MalZha1993}, when $f(\mathbf{x}) = \|\mathbf{x}\|_0$; and the Basis Pursuit (BP) algorithm \cite{CheDonSau1998,BruDonEla2009,elad2010}, when $f(\mathbf{x})= \|\mathbf{W}\mathbf{x}\|_1$, for example.

In the following sections we proceed to solve \cref{eq: compressed sensing for single pixel camera} and compare the results obtained by OMP and BP. We then demonstrate how we can use the sparsity index $\mathcal{E}$ in conjunction with the BP algorithm to decide when to increase the sampling rate, i.e., the number of measurements, to improve the reconstruction of the original image without prior knowledge of it.

\subsection{Solving the single pixel camera compressed sensing problem} \label{sec: solving the single pixel camera problem}
To define and solve \cref{eq: compressed sensing for single pixel camera}, we need to specify $\mathbf{P} \in \mathbb{R}^{M \times K}$ and $\mathbf{A} \in \mathbb{R}^{M \times N}$. Given an LCD of size $m \times n$ pixels, we note that it is customary to partition an image in smaller image blocks for individual processing, as is done in JPEG image compression \cite{TauMar2002,aust2008}, and therefore for our experiments we will partition the image in $l \times l$ image blocks. With this setting, from \cref{sec: compressed sensing details}, we have that $M = l \times l = l^2$, and $K \leq M \leq N$. 

 Let $\mathbf{e}_{i,j} \in \mathbb{R}^{l \times l}$ be the basis element matrix for $\mathbb{R}^{l \times l}$ with a value of one at column $i$ and row $j$, and zeros everywhere else. Let $\mathbf{A}_{i,j} = idct2(\mathbf{e}_{i,j}) \in \mathbb{R}^{l \times l}$ correspond to the 2D inverse discrete cosine transform matrix of $\mathbf{e}_{i,j}$, and let $\mathbf{a}_{i,j} \in \mathbb{R}^N$ be the vector that results from stacking the columns of $\mathbf{A}_{i,j}$. Finally, traverse the indices $i$ and $j$ in column-major-order, and define $\mathbf{A} = (\mathbf{a}_{i,j})$ as the matrix with column vectors $\mathbf{a}_{i,j}$, in that order. This setup implies that $N = l^2$, and therefore $M = N$. This gives an invertible $\mathbf{A} \in \mathbb{R}^{N \times N}$, putting us in the context of transform coding. \cref{fig: dct2 basis elements} shows in row-major-order the columns of $\mathbf{A}$ reshaped as $8  \times 8$ image blocks when setting $l = 8$. In this case, these blocks correspond to the basis elements used in the JPEG standard \cite{TauMar2002,aust2008}.

\begin{figure}[htb]
\centering
\includegraphics[width=.4\textwidth]{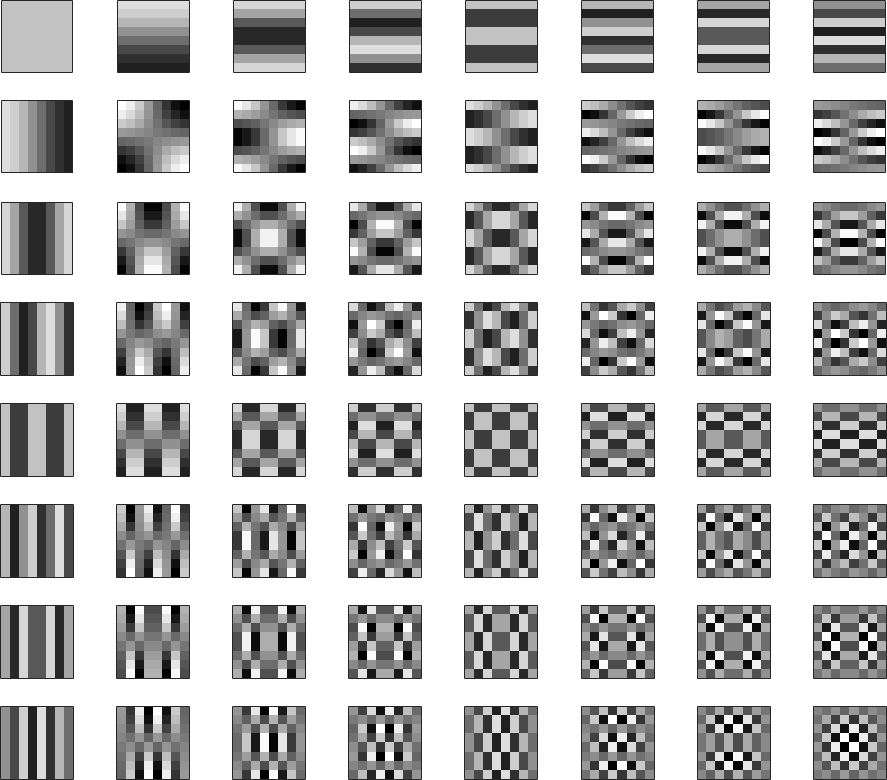}
\caption{2D discrete cosine transform analysis basis for $8 \times 8$ image blocks, as used in the JPEG standard.}
\label{fig: dct2 basis elements}
\end{figure}

Regarding the choice of $\mathbf{P} \in \mathbb{R}^{N \times K}$, we set $\mathbf{P}$ to be a matrix with entries chosen at random from $\{0,1\}$, i.e., $\mathbf{P}$ is the realization of a matrix whose entries are samples of the uniform probability distribution on $\{0,1\}$, denoted by $\mathcal{U}\{0,1\}$. Note that this choice is not entirely capricious as it relates to the experiment at hand and its corresponding modeling as described in \cref{sec: experiment}. With $\mathbf{P}$ set as above, we are choosing to have any one pixel element of the LCD to be either ``fully light transmitting" (value of 1) or ``fully light blocking" (value of 0). However, given that an LCD can have partial light transmission per pixel, if there were 256 possible uniform step values in such LCD, we could instead have picked the entries of $\mathbf{P}$ to be chosen at random from $\mathcal{U}\left\{0, \frac{1}{255}, \frac{2}{255}, \ldots, 1\right\}$, for example, but we stick with our choice of $\mathbf{P}$ above for the remainder of this work.

Now that we have established what $\mathbf{P}$ and $\mathbf{A}$ are, we are ready to briefly describe well known algorithms that attempt to solve \cref{eq: compressed sensing for single pixel camera} and which, under certain provable conditions, will converge to $\mathbf{x}^\star$, a solution of \cref{eq: compressed sensing for single pixel camera}. We will mention these conditions when presenting each algorithm in \cref{sec: omp,sec: bp}.

\subsubsection{Orthogonal Matching Pursuit (OMP)} \label{sec: omp}
The Orthogonal Matching Pursuit (OMP) algorithm was first presented in \cite{PatRezKri1993} and has since been refined a variety of times by multiple authors. This algorithm aims to find a solution $\mathbf{x}^\star$ to \cref{eq: compressed sensing for single pixel camera} when we set $f(\mathbf{x}) = \|\mathbf{x}\|_0$, i.e., when we want to solve for the vector with the smallest $\ell_0$-norm that explains the measurement vector $\mathbf{y}$. The basic idea of this algorithm is to iteratively add a nonzero entry to the previous vector in the iteration, starting with the zero vector, in a way that at each iteration the residual error, $\mathbf{r}^k = \mathbf{y} - \mathbf{P}^\text{T}\mathbf{A}\mathbf{x}^k$, is as small as possible. We reproduce the version found in \cite{BruDonEla2009} for convenience in \cref{alg: omp}.

\begin{algorithm}
  \caption{Orthogonal Matching Pursuit.
    \label{alg: omp}}
\begin{algorithmic}[0] 
\Require{A full rank matrix $\Phi$, a measurement vector $\mathbf{y}$, and an error threshold $\epsilon \geq 0$.}
\Ensure{A vector $\mathbf{x}^k$ such that $\| \mathbf{y} - \Phi \mathbf{x}^k\|_2 \leq \epsilon$, with $k \leq rank(\Phi)$.}
	\Statex
	\Function{OMP}{$\Phi, \mathbf{y}, \epsilon$}
	\State {\em Initialization.} $k \gets 0$, $\mathbf{x}^0 \gets \mathbf{0}$, $\mathbf{r}^0 \gets \mathbf{y} - \Phi \mathbf{x}^0 = \mathbf{y}$, $\mathcal{S}^0 \gets \emptyset$.
	\Repeat
	\State $k \gets k + 1$.
	\State {\em Sweep.} Compute the errors $\epsilon(j) = \min_{z_j} \| z_j \phi_j - \mathbf{r}^{k-1}\|_2$ for all $j$ using the optimal choice $z_j^\star = \phi_j^\text{T}\mathbf{r}^{k-1}/\|\phi_j\|_2^2$. \Comment{$\phi_j$ is the $j$th column of $\Phi$.}
	\State {\em Update support.} Find a minimizer $j_0$ of $\epsilon(j)$: $\forall\ j \notin \mathcal{S}^{k-1},\ \epsilon(j_0) \leq \epsilon(j)$, and update $\mathcal{S}^k = \mathcal{S}^{k-1} \cup \{j_0\}$.
	\State {\em Update provisional solution.} Compute $\mathbf{x}^k$, the minimizer of $\|\mathbf{y} - \Phi \mathbf{x}\|_2$ subject to $Support(\mathbf{x}) = \mathcal{S}^k$. \Comment{$Support(\mathbf{x}) = \left\{ j \in \mathbb{N} : x_j \neq 0, \mathbf{x} = (x_1, x_2, \ldots, x_N)^\text{T}\right\}$.}
	\State {\em Update residual.} Compute $\mathbf{r}^k = \mathbf{y} - \Phi \mathbf{x}^k$.
	\Until{$\|\mathbf{r}^k\|_2 \leq \epsilon$.}
	\State \Return{$\mathbf{x}^k$}
	\EndFunction
\end{algorithmic}
\end{algorithm}

A couple of remarks are in order. Solving \cref{eq: compressed sensing for single pixel camera} is NP-hard \cite{nata1995}, so it comes with no surprise that we can find examples of measurement vectors $\mathbf{y} \in \mathbb{R}^K$ where a solution $\mathbf{x}^\star$ to \cref{eq: compressed sensing for single pixel camera} for $f(\mathbf{x}) = \|\mathbf{x}\|_0$ satisfies $\|\mathbf{x}^\star\|_0 < K$, yet $\mathbf{x}_{omp} = \text{OMP}(\mathbf{P}^\text{T}\mathbf{A},\mathbf{y},0)$ will be such that $\|\mathbf{x}_{omp}\|_0 = K$. This follows from the readily verifiable fact that \cref{alg: omp} is of polynomial order whereas OMP is NP-hard, as mentioned before.

In \cite{BruDonEla2009}, for example, we can find conditions that guarantee when OMP will find an actual solution $\mathbf{x}^\star$ to \cref{eq: compressed sensing for single pixel camera} for $f(\mathbf{x}) = \|\mathbf{x}\|_0$, as well as when $\mathbf{x}^\star$ will be unique. For brevity we refer to that text for details.

\subsubsection{Basis Pursuit (BP)} \label{sec: bp}
Solving \cref{eq: compressed sensing for single pixel camera} when $f(\mathbf{x})= \|\mathbf{W}\mathbf{x}\|_1$ is called Basis Pursuit (BP) \cite{BruDonEla2009}. The matrix $\mathbf{W}$ gives all column vectors $\{\mathbf{d}_i\}_{i=1}^N$ of $\mathbf{P}^\text{T}\mathbf{A}$ an equal weight. It is a diagonal matrix whose $i$th diagonal element is given by $w_i = \|\mathbf{d}_i\|_2$. Without it, columns with larger $\ell^2$-norms tend to be penalized and their coefficients set to zero, or very small, biasing the solution. Notice that $f(\mathbf{x})= \|\mathbf{W}\mathbf{x}\|_1$ is a convex function, and therefore the vast literature on convex optimization with constraints can be brought fore to solve \cref{eq: compressed sensing for single pixel camera} in this case. Our particular approach to solving the compressed sensing problem under these circumstances is detailed in \cref{sec: convex minimization with constraints}.

\section{Numerical experiments and discussion} \label{sec: discussion}
Having defined the problem of compressed sensing in the context of images in \cref{sec: compressed sensing}, we are ready to test the sparsity index $\mathcal{E}$, see \cref{def: sparsity index}, derived from our new notion of sparsity $s$, see \cref{def: sparsity s}, and we show how $\mathcal{E}$ can be used to predict the quality of the reconstruction of an image via compressed sensing without prior knowledge of the original.

This is how we conduct our experiments. Given an image defined by its $\{ v_{i,j}\}$ pixel values, we  assume that we can subdivide it in blocks of size $l \times l$, where $l$ is to be defined shortly. If the original image's width or height were not divisible by $l$, we could extend the image to the right and bottom in a way that is similar to what is done in \cite{Kap1997}, for example, to make them both divisible by $l$. Then we process each image block, say $\mathbf{V}_k$, by first stacking from left to right its columns of pixel values to obtain a vector $\mathbf{v}_k \in \mathbb{R}^N$, where $N = l^2$ in this case. With the matrices $\mathbf{P} \in \mathbb{R}^{N \times K}$ and $\mathbf{A} \in \mathbb{R}^{N \times N}$ as defined in \cref{sec: solving the single pixel camera problem}, we proceed to obtain a measurement vector $\mathbf{y}_k = \mathbf{P}^\text{T}\mathbf{v}_k$ and solve for $\mathbf{x}^\star_k$ in \cref{eq: compressed sensing for single pixel camera} for either $f(\mathbf{x}) = \|\mathbf{x}\|_0$ or $f(\mathbf{x})= \|\mathbf{W}\mathbf{x}\|_1$, with OMP or BP, respectively. We need to define the value of $K$, and so we choose arbitrarily to set $K=0.75 N$, that is $\frac{3}{4}$ the number of measurements needed to directly reconstruct the image vector $\mathbf{v}_k$ from the measurement vector $\mathbf{y}_k$. Finally, we obtain the compressed sensing reconstruction image block $\widetilde{\mathbf{V}}_k$ by stacking into $l$ columns of length $l$ the sequential values of the reconstruction vector $\widetilde{\mathbf{v}}_k = \mathbf{A}\mathbf{x}^\star_k$. After processing and putting in place all the blocks in their respective order, we compare the original image with its reconstruction $\{\widetilde{v}_{i,j}\}$, possibly trimming first from both images any extensions to the right and bottom that we might have added to the original to make its dimensions divisible by $l$.

In preliminary experiments we compared the performance of BP for $l = 8, 16, 32, \text{and } 64$, and we found that $l=32$ gives both the fastest and best reconstruction---as measured in seconds and dB for PSNR, respectively. Hence we set $l=32$. This results in having $N=1024$, and $K=768$. For $l=32$, the time it takes to process an image block with OMP to solve \cref{eq: compressed sensing for single pixel camera} when $f(\mathbf{x}) = \|\mathbf{x}\|_0$ is just too long to make it a practical method, not to mention that the reconstruction quality obtained with OMP is worse both in PSNR and Mean Structural Similarity Index (MSSIM) \cite{WanBovSheSim2004} than the one obtained with BP. Therefore we conducted all of our experiments exclusively with BP. See \cref{fig: reconstruction comparison}.

\begin{figure}[htbp]
  \centering
  \subfloat[Original]{\label{fig: original}\includegraphics[width=.3\linewidth]{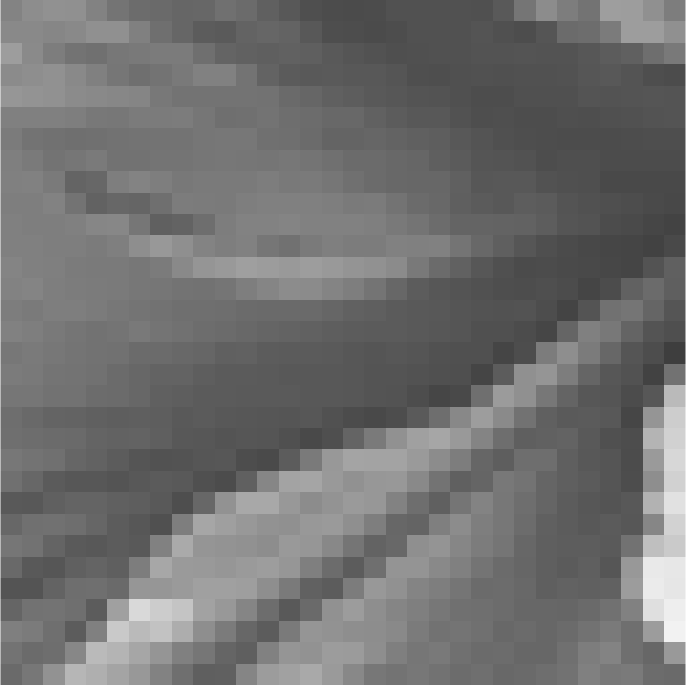}}
  \qquad
  \subfloat[BP reconstruction]{\label{fig: bp reconstruction}\includegraphics[width=.3\linewidth]{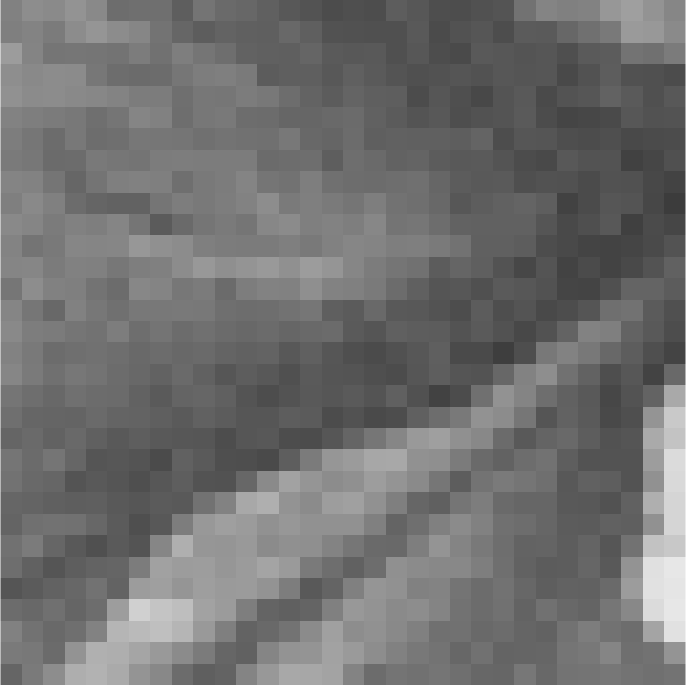}}
  \qquad
  \subfloat[OMP reconstruction]{\label{fig: omp reconstruction}\includegraphics[width=.3\linewidth]{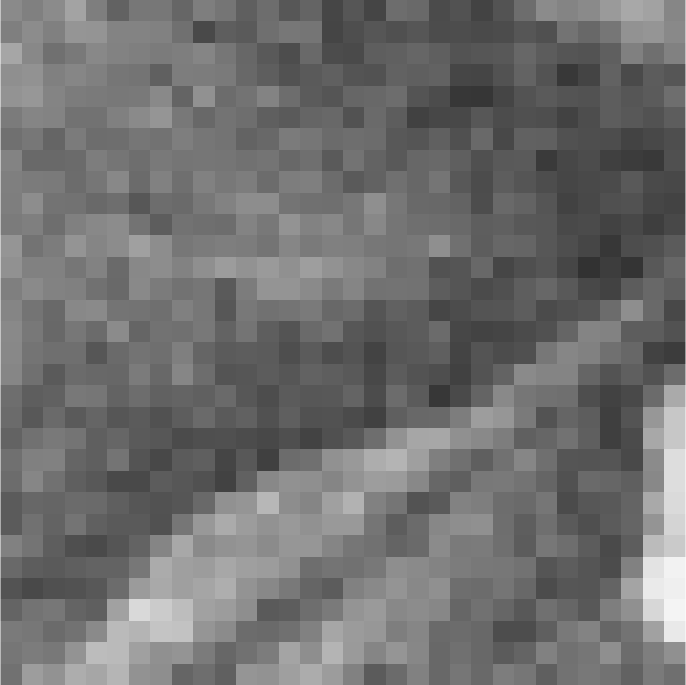}}
  \caption{Reconstruction from 75\% of the required samples for perfect reconstruction of a $32 \times 32$ pixel image by means of \protect\subref{fig: bp reconstruction} basis pursuit (BP), with errors of 31.970006 dB in PSNR and 0.843123 in MSSIM; and \protect\subref{fig: omp reconstruction} orthogonal matching pursuit (OMP), with errors of 26.976103 dB in PSNR and 0.671464 in MSSIM. The BP method is better and faster than OMP by 4.993903 dB, with an average speedup of 8.14, when processing \protect\subref{fig: original}. The parameters for the line search, \cref{alg: backtracking line search}, were set to $\alpha = 0.001$ and $\beta = 0.75$. The entries of the sampling matrix $\mathbf{P} \in \{0,1\}^{1024 \times 768}$ are drawn at random from $\mathcal{U}\{0,1\}$.}
\label{fig: reconstruction comparison}
\end{figure}

The data set of images that we used corresponds to the luminance of the first 24 reference images in the {\em TID2008} image database \cite{PonLukZelEgiCarBat2009}. Each original image is a color image of 512 pixels wide by 384 pixels tall, $384 \times 512$ pixels in matrix form. This size is convenient for our purposes since it results in 192 disjoint $32 \times 32$ image blocks, and therefore we don't need to extend the original images in any form for blocking and processing. For our experiments, we combined the color information of each image block into a single vector by computing its luminance. The luminance $Y$ of an image is computed from its red $R$, green $G$, and blue $B$ channels with the formula \cite{TauMar2002,aust2008},
\begin{equation*}
Y = 0.299\, R + 0.587\, G + 0.114\, B.
\end{equation*}

We observe that by the choice of $\mathbf{A}$ that we have made, the solution $\mathbf{x}^\star_k$ obtained either by BP or OMP for a given measurement vector $\mathbf{y}_k$ has as coordinates all of the 2D DCT coefficients of the respective reconstruction block $\widetilde{\mathbf{V}}_k$, and therefore, by \cref{def: sparsity index}, the sparsity index of $\widetilde{\mathbf{V}}_k$ is given by
\begin{equation*}
\mathcal{E}(\widetilde{\mathbf{V}}_k) = \frac{s(\mathbf{x}^\star_k)}{N}.
\end{equation*}

If we pay attention for a moment to the solution $\mathbf{x}^\star_k$ given by OMP, we observe that we must necessarily have $\|\mathbf{x}^\star_k\|_0 \leq K$. This is easily verifiable by the design of OMP, see \cref{alg: omp}, and the fact that $rank(\mathbf{P}^\text{T}\mathbf{A}) = K$. Then, using the terminology from \cref{sec: error analysis and sparsity}, we must have that the number $i_0$ of entries equal to zero in $\mathbf{x}^\star_k$ must satisfy $i_0 = N - \|\mathbf{x}^\star_k\|_0 \geq N-K > 0$. Recalling the definitions of $i_\star$ and $p_\star$ in \cref{theo: unique p_star}, if $i_0 < i_\star = i_\star(\mathbf{x}_k,p_\star)$, where $\mathbf{x}_k$ is the vector that results from stacking the columns of the 2D DCT transform of $\mathbf{V}_k$, then the results from \cref{sec: error analysis and sparsity} apply and we can estimate a lower bound for $PSNR(\mathbf{V}_k,\widetilde{\mathbf{V}}_k)$ provided $\mathbf{x}_k$ and $\mathbf{x}^\star_k$ are {\em close enough}, by substituting in \cref{eq: psnr lower bound detailed calculation} $\mathbf{x}_k$ with $\mathbf{x}^\star_k$, and $s(\mathbf{x}_k)/N = \mathcal{E}(\mathbf{V}_k)$ with $s(\mathbf{x}^\star_k)/N = \mathcal{E}(\widetilde{\mathbf{V}}_k)$. Under these assumptions, we must have that a lower bound for $PSNR(\mathbf{V}_k,\widetilde{\mathbf{V}}_k)$ is then approximated by
\begin{multline} \label{eq: lower bound approximation}
-\frac{20}{p_\star}\log_{10}\left(\mathcal{E}(\widetilde{\mathbf{V}}_k)\right)
+ 10\frac{p_\star - 2}{p_\star} \log_{10}(N)
+ \frac{20}{p_\star} \log_{10}(N - i_\star +1) \\
+ 20 \log_{10}(MAX)
- 10 \log_{10}(N - \|\mathbf{x}^\star_k\|_0).
\end{multline}

We focus our attention on the first term of \cref{eq: lower bound approximation}, $-20\log_{10}(\mathcal{E}(\widetilde{\mathbf{V}}_k))/p_\star$,  by noting that since $0 \leq \mathcal{E}(\widetilde{\mathbf{V}}_k) \leq 1$, the smaller the value of $\mathcal{E}(\widetilde{\mathbf{V}}_k)$, the larger the lower bound of $PSNR(\mathbf{V}_k,\widetilde{\mathbf{V}}_k)$ will be, provided all other variables are held constant. Notice also that the estimate in \cref{eq: lower bound approximation} uses information obtained exclusively from the solution $\mathbf{x}^\star_k$ of the compressed sensing problem \cref{eq: compressed sensing for single pixel camera}, which is derived from the measurement vector $\mathbf{y}_k$. That is, no knowledge of the original image block $\mathbf{V}_k$ is required, except for the partial information derived from it by means of our sampling matrix $\mathbf{P}$.

This analysis leads us to propose the following statistical hypothesis. We claim that, with high probability, whenever the sparsity index $\mathcal{E}$ is less than or equal to a certain threshold $t_0$, the peak signal to noise ratio of the reconstruction will be greater than or equal to a given decibel value $r_0$, and vice versa. This proposition is equivalent to its converse hypothesis $H$,
\begin{equation} \label{eq: hypothesis}
H : \mathcal{E}\big(\widetilde{\mathbf{V}}_k\big) > t_0 \Leftrightarrow PSNR\big(\mathbf{V}_k,\widetilde{\mathbf{V}}_k\big) < r_0,
\end{equation}
where $t_0$ is the threshold that corresponds to a PSNR of $r_0$ dB for a particular reference image block. We set $t_0 = \mathcal{E}(\widetilde{\mathbf{V}}_0) = 0.734166$ and $r_0 = PSNR(\mathbf{V}_0,\widetilde{\mathbf{V}}_0) = 31.970006$ dB, where $\widetilde{\mathbf{V}}_0$ is the reconstruction obtained with BP from the measurements taken from the image block $\mathbf{V}_0$, the $32 \times 32$ image block in I07.BMP with top-left coordinates $(225,129)$. This image block is shown in \cref{fig: original}. We choose it as our reference image block because its reconstruction errors are close to 32 dB in PSNR, and 0.85 in MSSIM, which are both acceptable for images.

\begin{figure}[htbp]
  \centering
  \subfloat[Original (I07.BMP)]{\label{fig: i07orig}\includegraphics[width=.3\linewidth]{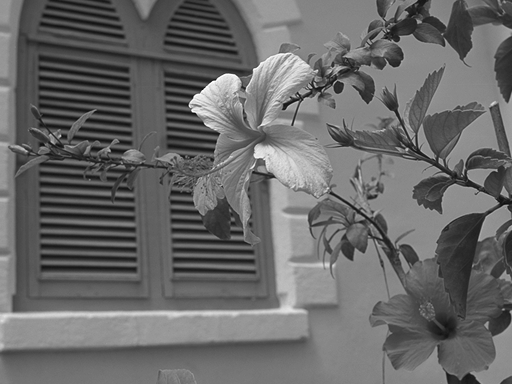}}
  \
  \subfloat[BP reconstruction]{\label{fig: i07rec}\includegraphics[width=.3\linewidth]{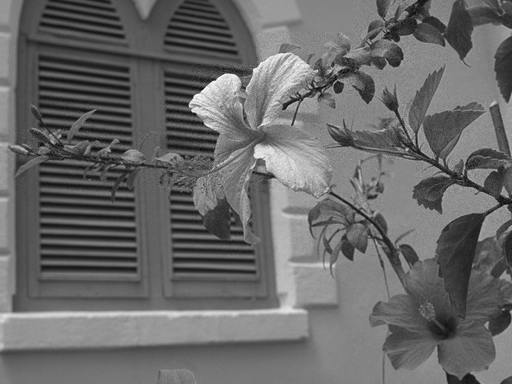}}
  \
  \subfloat[$SSIM$ map]{\label{fig: i07ssim}\includegraphics[width=.3\linewidth]{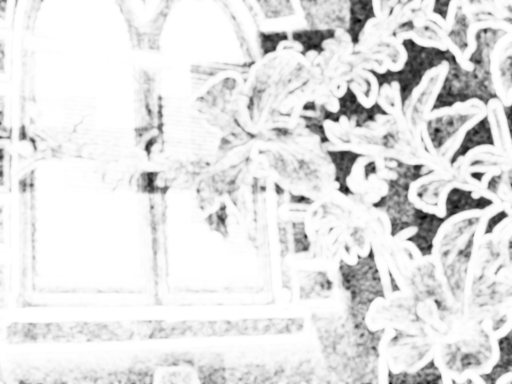}}
  \linebreak
  \subfloat[$\mathcal{E}$]{\label{fig: 07-E}\includegraphics[width=.3\linewidth]{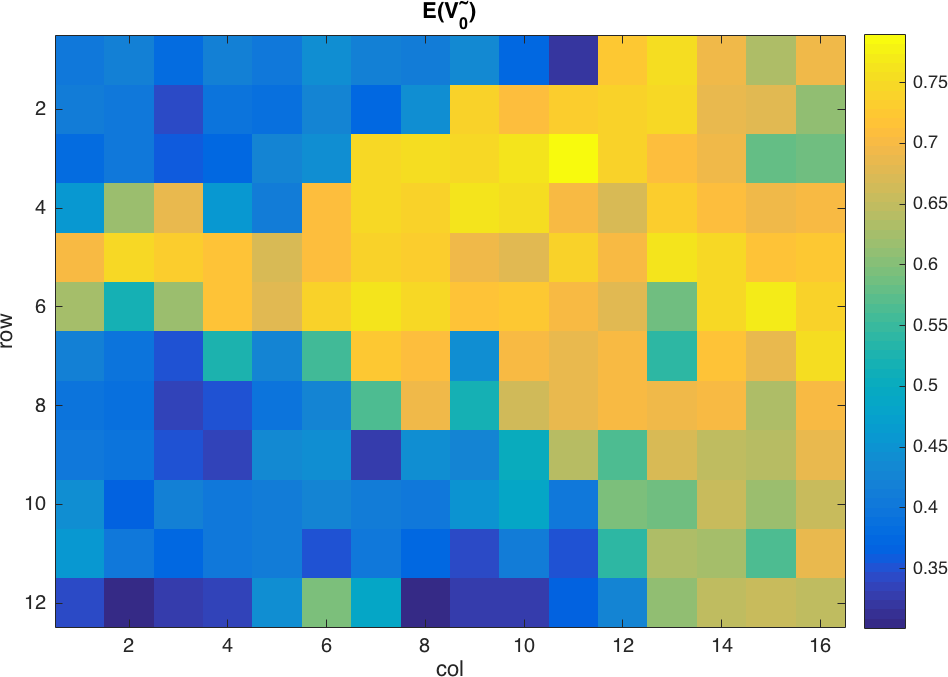}}
  \
  \subfloat[$PSNR$]{\label{fig: 07-psnr}\includegraphics[width=.3\linewidth]{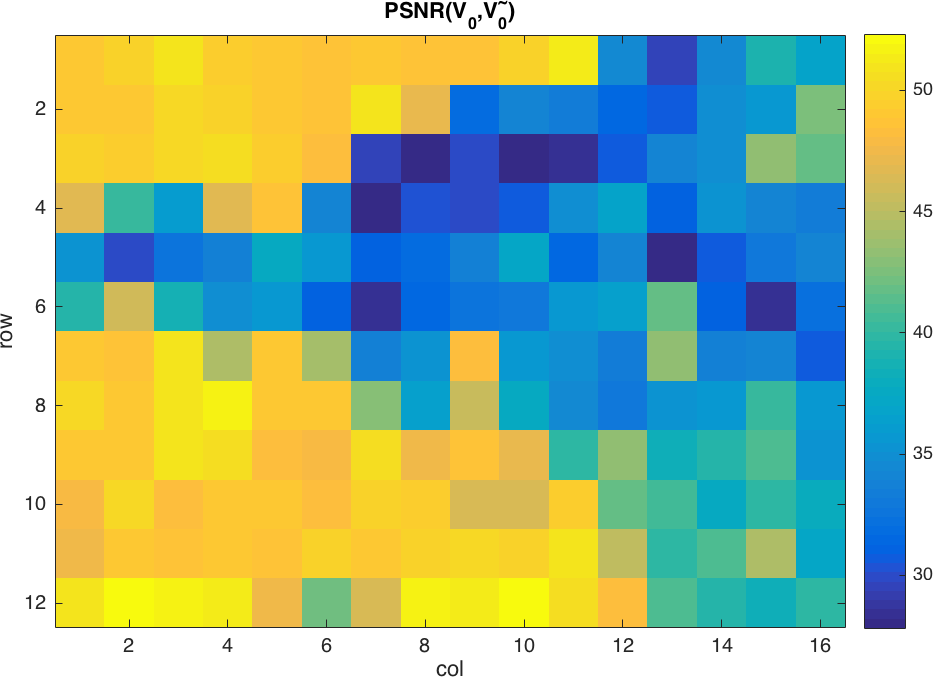}}
  \linebreak
  \subfloat[$\mathcal{E}>t_0$]{\label{fig: 07-E>}\includegraphics[width=.3\linewidth]{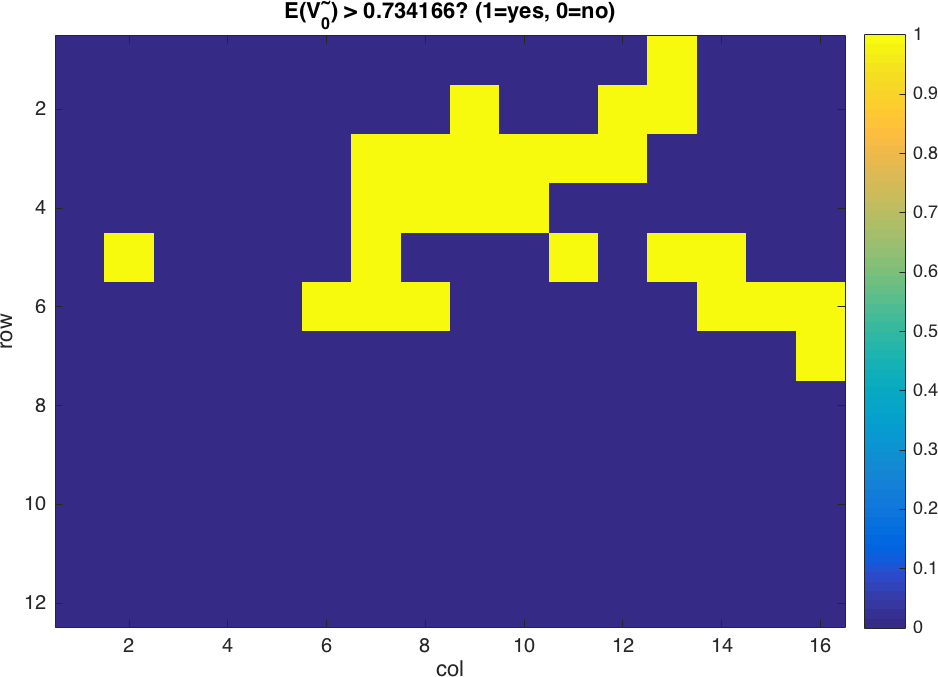}}
  \
  \subfloat[$PSNR < r_0$]{\label{fig: 07-psnr<}\includegraphics[width=.3\linewidth]{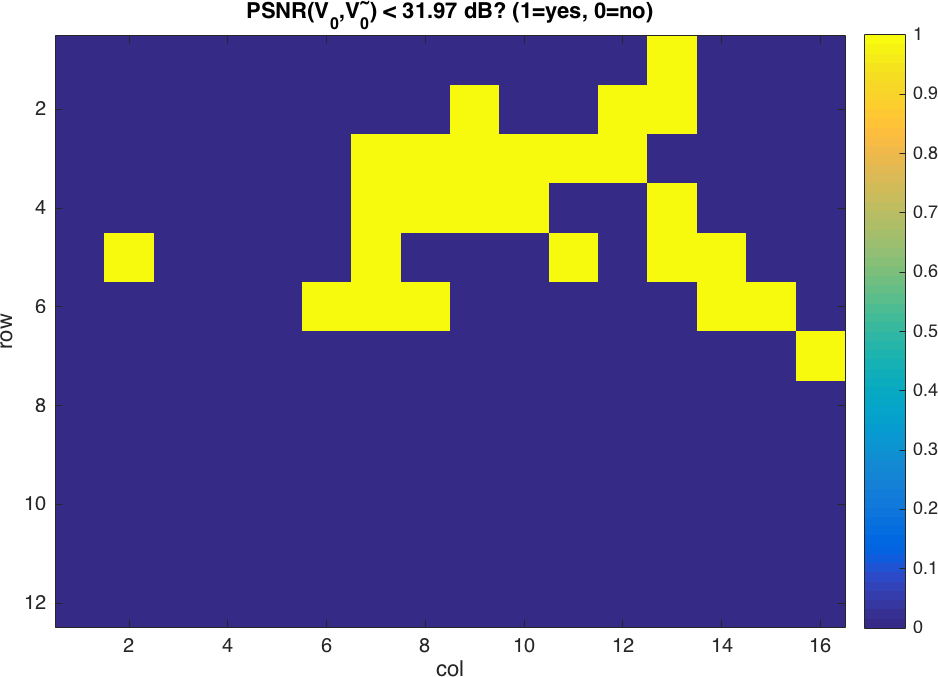}}
  \
  \subfloat[$H : \mathcal{E} > t_0 \Leftrightarrow PSNR < r_0$]{\label{fig: 07-fpfn}\includegraphics[width=.3\linewidth]{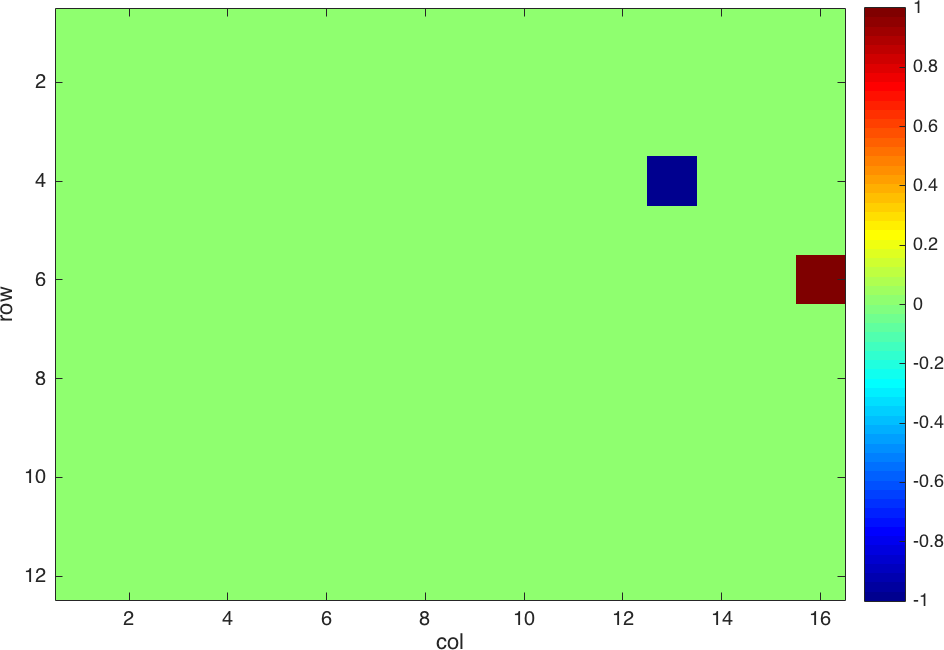}}
  \caption{Experimental treatment of an image. See the text for details.}
\label{fig: full figure treatment}
\end{figure}

In \cref{fig: full figure treatment} we show the full experimental treatment of \cref{fig: i07orig}, the luminance of figure I07.BMP in the {\em TID2008} database, as we described it in the beginning of this section. \cref{fig: i07rec} shows the reconstruction from BP after processing all 192 disjoint image blocks of size $32 \times 32$ that constitute the full image. \cref{fig: i07ssim} shows the {\em structural similarity map} (SSIM) of the reconstruction, where lighter pixels mean better reconstruction and darker worse, see \cite{WanBovSheSim2004}. \cref{fig: 07-E,fig: 07-psnr} show the values of the sparsity index $\mathcal{E}(\widetilde{\mathbf{V}}_k)$ and the peak signal-to-noise ratio $PSNR(\mathbf{V}_k,\widetilde{\mathbf{V}}_k)$ for each of the 192 image blocks, respectively. Notice, as we expected, the inverse relationship between both quantities. In \cref{fig: 07-E>} we have applied our threshold $t_0$ to the map shown in \cref{fig: 07-E}. The image shows in yellow the image blocks for which $\mathcal{E}(\widetilde{\mathbf{V}}_k) > t_0$ and in blue those for which the inequality doesn't hold. Similarly for \cref{fig: 07-psnr<}, we compare the values of the map in \cref{fig: 07-psnr} to our threshold $r_0$, except that in this case the yellow image blocks correspond to the cases when $PSNR(\mathbf{V}_k,\widetilde{\mathbf{V}}_k) < r_0$, and blue when this inequality doesn't hold. Finally, \cref{fig: 07-fpfn} shows in green the image blocks for which hypothesis $H$ is validated, in red those for which we have made a false positive, or error of type I, and in blue those for which we have made a false negative prediction, or error of type II. 

\begin{table}[tbhp]
{\footnotesize
  \caption{Number of false positive (type I) and false negative (type II) errors by image.} \label{tab: fpfn}
\begin{center}
  \begin{tabular}{|c|c|c|c|c|c|} \hline
   Image & \bf \# type I & \bf \# type II & \bf \# type I+II & PSNR (dB) & MSSIM \\ \hline
I01.BMP	&	4 (2.08\%)	&	16 (8.33\%)	&	20 (10.42\%)	& 30.153968 & 0.833950 \\
I02.BMP	&	0		&	4 (2.08\%)		&	4 (2.08\%)		& 36.415267 & 0.798317 \\
I03.BMP	&	2 (1.04\%)	&	4 (2.08\%)		&	6 (3.13\%)		& 38.186510 & 0.812061 \\
I04.BMP	&	2 (1.04\%)	&	0			&	2 (1.04\%)		& 39.694344 & 0.870662 \\
I05.BMP	&	6 (3.13\%)	&	15 (7.81\%)	&	21 (10.94\%)	& 29.734894 & 0.856745 \\
I06.BMP	&	3 (1.56\%)	&	12 (6.25\%)	&	15 (7.81\%)	& 30.984736 & 0.835197 \\
I07.BMP	&	1 (0.52\%)	&	1 (0.52\%)		&	2 (1.04\%)		& 35.966860 & 0.860078 \\
I08.BMP	&	9 (4.69\%)	&	15 (7.81\%)	&	24 (12.50\%)	& 29.623119 & 0.857145 \\
I09.BMP	&	1 (0.52\%)	&	4 (2.08\%)		&	5 (2.60\%)		& 35.340432 & 0.822734 \\
I10.BMP	&	2 (1.04\%)	&	5 (2.60\%)		&	7 (3.65\%)		& 35.416788 & 0.847086 \\
I11.BMP	&	2 (1.04\%)	&	11 (5.73\%)	&	13 (6.77\%)	& 31.287853 & 0.806722 \\
I12.BMP	&	2 (1.04\%)	&	2 (1.04\%)		&	4 (2.08\%)		& 36.634208 & 0.817618 \\
I13.BMP	&	4 (2.08\%)	&	7 (3.65\%)		&	11 (5.73\%)	& 26.320210 & 0.830284 \\
I14.BMP	&	3 (1.56\%)	&	17 (8.85\%)	&	20 (10.42\%)	& 31.304727 & 0.852535 \\
I15.BMP	&	3 (1.56\%)	&	5 (2.60\%)		&	8 (4.17\%)		& 36.433619 & 0.798157 \\
I16.BMP	&	2 (1.04\%)	&	1 (0.52\%)		&	3 (1.56\%)		& 36.064543 & 0.856455 \\
I17.BMP	&	0		&	0			&	0			& 34.165445 & 0.867691 \\
I18.BMP	&	2 (1.04\%)	&	5 (2.60\%)		&	7 (3.65\%)		& 30.604079 & 0.851522 \\
I19.BMP	&	3 (1.56\%)	&	15 (7.81\%)	&	18 (9.36\%)	& 33.746116 & 0.876263 \\
I20.BMP	&	1 (0.52\%)	&	4 (2.08\%)		&	5 (2.60\%)		& 34.260744 & 0.737556 \\
I21.BMP	&	2 (1.04\%)	&	6 (3.13\%)		&	8 (4.17\%)		& 30.994745 & 0.818269 \\
I22.BMP	&	5 (2.60\%)	&	11 (5.73\%)	&	16 (8.33\%)	& 33.095956 & 0.818456 \\
I23.BMP	&	1 (0.52\%)	&	0			&	1 (0.52\%)		& 37.067961 & 0.853895 \\
I24.BMP	&	2 (1.04\%)	&	8 (4.17\%)		&	10 (5.21\%)	& 32.519183 & 0.831541 \\ \hline
  \end{tabular}
\end{center}
}
\end{table}

\cref{tab: fpfn} summarizes the results for all 24 natural images in {\em TID2008}. In it, we present the number of false positive (type I) and false negative (type II) errors per image, as well as the total number of errors of either type, and the overall PSNR and MSSIM errors for each BP image reconstruction. These results suggest that, with high probability, hypothesis $H$, see \cref{eq: hypothesis}, is true.

It is worth noting that errors of type I are more benign than those of type II. This is because we are interested in finding out when our image block reconstruction is going to be below a certain quality threshold. That is, if for a given image block reconstruction $\widetilde{\mathbf{V}}_k$ we were to compute that $\mathcal{E}(\widetilde{\mathbf{V}}_k) > t_0$, under the assumption of the validity of hypothesis $H$, we would conclude that our reconstruction has a poor peak signal-to-noise ratio, i.e., $PSNR(\mathbf{V}_k,\widetilde{\mathbf{V}}_k) < r_0$. But, if we have made an error of type I, the image block reconstruction is in fact such that $PSNR(\mathbf{V}_k,\widetilde{\mathbf{V}}_k) \geq r_0$, i.e., better than predicted by $H$.

The opposite is true of errors of type II, which lead us to believe that we have done better than we actually did.

In \cref{fig: E vs PSNR} we show in two different formats the graph of the sparsity index $\mathcal{E}(\widetilde{\mathbf{V}}_k)$ versus the peak signal-to-noise ratio $PSNR(\mathbf{V}_k,\widetilde{\mathbf{V}}_k)$ for all 192 image blocks of each of the 24 natural images in {\em TID2008} that we studied, a total of 4608 image blocks. The dot at the intersection of the red, blue, and green regions in \cref{fig: EvsPSNR2} corresponds to the intersection of the horizontal and vertical lines with common point $(t_0,r_0)$ that we chose for our experiments, $t_0 = 0.734166$ and $r_0 = 31.970006 \text{ dB}$. If we number the quadrants that these two lines define, numbering them clockwise starting in the top right position, quadrant 1 contains all red dots for which we have a false positive, or error of type I, i.e., image blocks that have sparsity index $\mathcal{E}$ greater than $t_0$ but for which their PSNR is above $r_0$; quadrants 2 and 4, which contain all of the green dots that validate hypothesis $H$; and finally, quadrant 3 with the blue dots, which correspond to image blocks with errors of type II, namely, false negatives for which their sparsity index $\mathcal{E}$ is below $t_0$, yet their PSNR is below $r_0$.

In summary, in this work we have shown that by computing $\mathcal{E}(\widetilde{\mathbf{V}}_k)$ and comparing it to $t_0$, we can predict with a degree of certainty whether our compressed sensing image block reconstruction using BP will be satisfactory, or not. Therefore, for those image blocks for which $\mathcal{E}(\widetilde{\mathbf{V}}_k) > t_0$, we can decide to increase the number of measurements (samples) to improve their reconstruction, giving a predictive refinement methodology for compressed sensing imaging.

\begin{figure}[htp]
    \centering
    \subfloat[Molecular representation]{\label{fig: EvsPSNR}\includegraphics[width=0.49\textwidth]{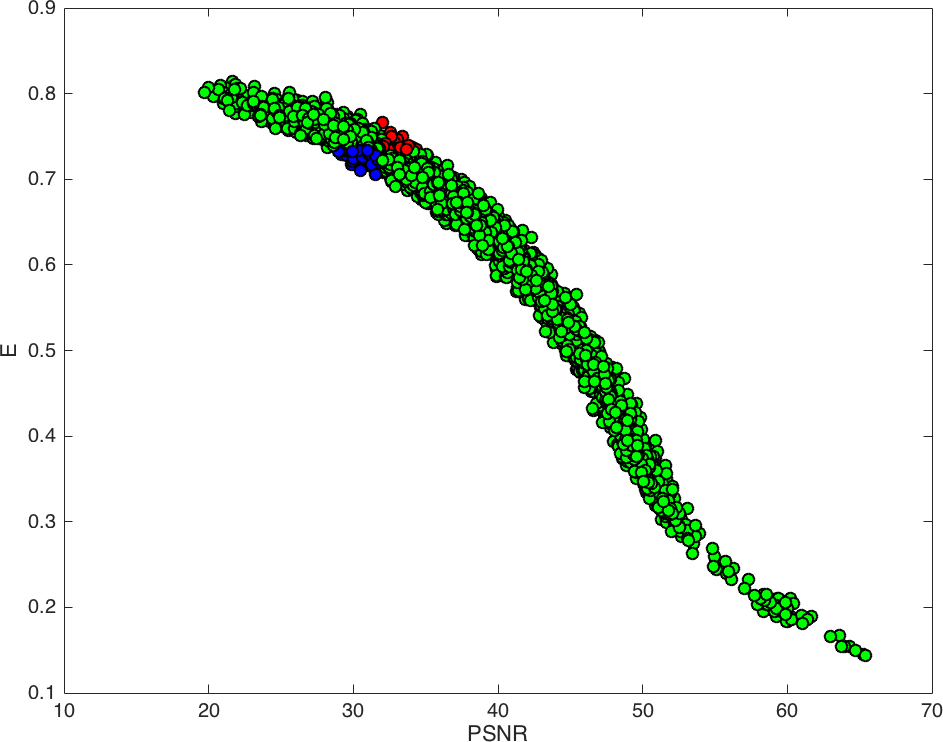}}
    \
    \subfloat[Density representation]{\label{fig: EvsPSNR2}\includegraphics[width=0.49\textwidth]{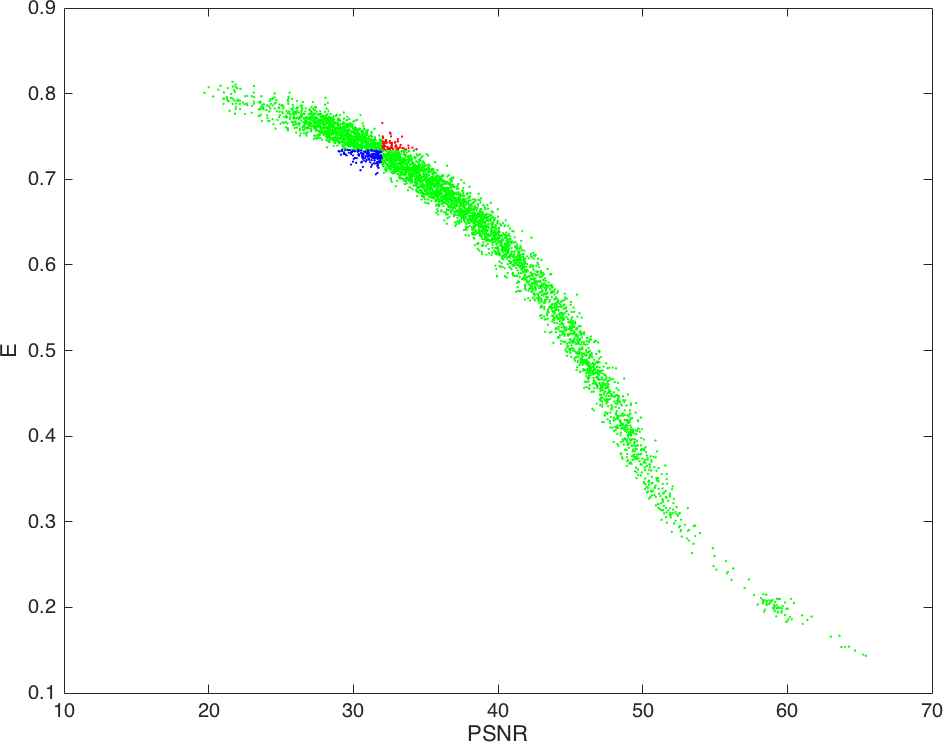}}
    \caption{$\mathcal{E}$ vs $PSNR$. Aggregate of all $32 \times 32$ image blocks of all images treated. The green dots correspond to image blocks for which hypothesis $H$ is valid, the red dots are image blocks with error of type I, and the blue dots are those with error of type II.}
    \label{fig: E vs PSNR}
\end{figure}

\appendix
\section{Sparsity, information, and the discrete cosine transform} \label{sec: raionale for E}
In this section we motivate the definition of the sparsity index $\mathcal{E}$, defined in \cref{sec: sparsity index}.

The main idea goes as follows. Suppose that we have a source of information, which is stochastic in nature, embodied in a vector $\mathbf{v} = (v_1, v_2, \ldots, v_n)^\text{T} \in \mathbb{R}^n$, where the $\{v_i\}$ are samples of some $\{V_i\}$ random variables, and suppose that we have a unitary transform $\mathbf{T} \in \mathbb{R}^{n \times n}$ such that if $\mathbf{x} = \mathbf{T}\mathbf{v}$, then the corresponding sample covariance matrix $C_\mathbf{x}$ is diagonal. Then, the sampled correlation coefficients $\rho_{x_i x_j}$ will be zero, and the linear portion of the mutual information of $x_i$ and $x_j$ will be zero as well---provided we can show a result that  links correlation and mutual information---in essence packing the linear portion of the mutual information into the individual components of $\mathbf{x}$. If we then quantify the {\em relevance} of the $\{x_i\}$ by way of the sparsity $s(\mathbf{x})$, we would have identified information with sparsity, which will allow us to say that the sparser $\mathbf{x}$ is, the less information it contains.  If you need a more concrete notion of what we mean by the {\em relevance} of a component $x_i$, we will say that $x_i$ is relevant if $|x_i| \geq 1$, and irrelevant if $|x_i| < 1$. Finally, if we had a result that said something about an invertible transform preserving mutual information, then we can also say something about the information content of $\mathbf{v}$ by way of the sparsity $s(\mathbf{x})$ of its transform.

We now review the concepts that we introduced above to link the notions of sparsity and information via the idea of {\em relevance} that is captured by our definition of sparsity $s$.

Let $X$ and $Y$ be two real discrete random variables. Let $X$ take on $n$ values, say $x_1, x_2, \ldots, x_n$ with respective probabilities $P_X(1), P_X(2), \ldots, P_X(n)$, and let $Y$ take on $m$ values $y_1, y_2, \ldots, y_m$ with respective probabilities $P_Y(1), P_Y(2), \ldots, P_Y(m)$. In general we assume that the variables $X$ and $Y$ are interdependent, and we denote $P_{XY}(i,j)$ the joint probability of $X$ and $Y$ taking on the values $x_i$ and $y_j$, respectively. With this setup the mutual information between $X$ and $Y$, measured in bits, is defined by Shannon \cite{shan1948} as
\begin{equation} \label{eq: mutual information}
I(X,Y) = \sum_{i=1}^n \sum_{j=1}^m P_{XY}(i,j) \log_2\left(\frac{P_{XY}(i,j)}{P_X(i)P_Y(j)}\right).
\end{equation}
This quantity, intuitively, measures how much information is gained of the value of one variable upon knowledge of the value of the other. Note that if $X$ and $Y$ are independent, then $I(X,Y) = 0$, since, by definition of independence, $P_{XY}(i,j) = P_X(i) P_Y(j)$.

Let $\mu_X = E[X] = \sum_{i=1}^n x_i P_X(i)$ and $\mu_Y = E[Y] = \sum_{j=1}^m y_j P_Y(j)$ denote the mean---also known as the expected value---of the random variables $X$ and $Y$, respectively. The covariance $cov(X,Y)$---also denoted as $\sigma_{XY}$---of $X$ and $Y$ is defined as
\begin{equation} \label{eq: covariance of two random variables}
\sigma_{XY} = E[(X - \mu_X)(Y - \mu_Y)] = E [XY] - \mu_X \mu_Y =
\left(\sum_{i=1}^n \sum_{j=1}^m x_i y_j P_{XY}(i,j)\right) - \mu_X \mu_Y.
\end{equation}

The correlation coefficient $\rho_{XY}$ of $X$ and $Y$ is defined as
\begin{equation} \label{eq: correlation coefficient}
\rho_{XY} = \frac{\sigma_{XY}}{\sigma_X \sigma_Y},
\end{equation}
whenever $\sigma_X := \sqrt{\sigma_{XX}}$ and $\sigma_Y := \sqrt{\sigma_{YY}}$, the standard deviations of $X$ and $Y$, respectively, are nonzero. The correlation coefficient measures the linear relationship between two random variables, with perfect linear increasing relationship if $\rho_{XY} = 1$, and perfect linear decreasing relationship if $\rho_{XY} = -1$. When $\rho_{XY} = 0$ we say that $X$ and $Y$ are uncorrelated. Note that if $X$ and $Y$ are independent then $\rho_{XY} = 0$, but the inverse is not true, as the correlation coefficient only measures linear dependence.

If $X$ and $Y$ are Gaussian random variables, then Gel'fand and Yaglom \cite{GelYag1957} prove that
\begin{equation} \label{eq: mutual information and correlation coefficient for two Gaussians}
I(X,Y) = -\frac{1}{2} \log_2\left(1 - \rho_{XY}^2\right).
\end{equation}
This is a remarkable result in that it links correlation with information, at least in the particular case of two Gaussian random variables. In this case note that if $X$ and $Y$ are uncorrelated, we don't gain any information about one variable by learning something of the other, which is consistent with how we intuitively think of correlation.

Before we proceed any further, note that the definition of mutual information in \cref{eq: mutual information} preserves its meaning if we change the real random variables $X$ and $Y$ for cartesian random vectors $\mathbf{x}$ and $\mathbf{y}$. With this generalization to higher dimensions of mutual information in mind, we cite another result in \cite{GelYag1957} relevant to our purposes.
\begin{theorem} \label{theo: mutual information and linear transforms}
Let $\mathbf{A} \in \mathbb{R}^{k \times k}$, and $\mathbf{x} \in \mathbb{R}^k$ be a random vector. Then
\begin{equation}
I(\mathbf{x},\mathbf{y}) \geq I(\mathbf{A}\mathbf{x},\mathbf{y})
\end{equation}
holds for any random vector $\mathbf{y}$, with equality if the matrix $\mathbf{A}$ is non-singular.
\end{theorem}
In particular note that \cref{theo: mutual information and linear transforms} states that we don't win nor lose any information from applying an invertible linear transformation to a random vector.

Let $\mu_{\mathbf{x}} = E[\mathbf{x}]$ denote the mean of the random vector $\mathbf{x}$. The covariance matrix $C_\mathbf{x}$ of $\mathbf{x}$ is then defined as 
\begin{equation} \label{eq: covariance matrix}
C_\mathbf{x} = E\left[(\mathbf{x} - \mu_\mathbf{x})(\mathbf{x} - \mu_\mathbf{x})^\text{T}\right]
= E\left[\mathbf{x} \mathbf{x}^\text{T}\right] - \mu_\mathbf{x} \mu_\mathbf{x}^\text{T}.
\end{equation}
Note that $C_\mathbf{x} = C_\mathbf{x}^\text{T}$ is a symmetric matrix whose element at row $i$ and column $j$ corresponds to the covariance $\sigma_{X_i X_j}$ of the random variables $X_i$ and $X_j$, if we set $\mathbf{x} = (X_1, X_2, \ldots, X_k)^\text{T}$. As expected, the covariance matrix gives information of the linear relationships that may exist among the components $\{X_i\}$ of the random vector $\mathbf{x}$.


The discrete cosine transform (DCT) \cite{AhmNatRao1974}, also known as the DCT-II---the version that MATLAB implements---is one of several discrete cosine transforms \cite{RaoYip1990}. It is a unitary transform that asymptotically approximates the Karhunen--Lo\`eve transform (KLT) of the input signal, provided the signal can be modeled as a first-order Markov process. The KLT decorrelates a finite discrete signal optimally in the sense that the $MSE$ between the signal and its truncated representation in the KLT basis that eliminates all but the $m$ largest coefficients in absolute value, is minimal. Hence, the energy compaction performance of the DCT approaches that of the KLT when the input signal is as assumed above \cite{RaoYip1990,reed1993}.

The KLT for a given signal vector $\mathbf{x}$ is driven by the process of diagonalizing the covariance matrix $C_\mathbf{x}$, a procedure that in essence eliminates the linear correlation among the entries of $\mathbf{x}$, while preserving their mutual information since the KLT is invertible, see \cref{theo: mutual information and linear transforms}.

The magnitude of the KLT coefficients of a given signal can then be associated with how important the corresponding KL basis element is, how much weight it should be given in conveying information about the signal. The DCT-II being an approximation of the KLT, we can also think of the magnitude of the DCT-II coefficients as representing how much of the information in the signal they carry.

Therefore, under this setup, the sparsity index $\mathcal{E}$ can be interpreted as a measure of information by a count of the relevant coefficients that describe a signal, linking sparsity with information as Shannon defined it. This fulfills our goal to motivate the definition of $\mathcal{E}$.

\section{Convex minimization with linear constraints} \label{sec: convex minimization with constraints}
Given a vector $\mathbf{y} \in \mathbb{R}^m$, consider the problem,
\begin{equation} \label{eq: equality constrained minimization}
\min_{\mathbf{x} \in \mathbb{R}^n} f(\mathbf{x}) \quad \text{subject to } \mathbf{A}\mathbf{x} = \mathbf{y},
\end{equation}
where $f : \mathbb{R}^n \rightarrow \mathbb{R}$ is a convex and continuously differentiable function, and $\mathbf{A} \in \mathbb{R}^{m \times n}$ with $rank(\mathbf{A}) = m < n$.

\Cref{eq: equality constrained minimization} defines an {\em equality constrained minimization problem}, which has been studied extensively, see for example \cite{NasSof1996,BoyVan2004,GriNasSof2008}. Here, we present the approach that we have followed to solve \cref{eq: equality constrained minimization} to obtain the results in this work. 

\subsection{Constrained descent methods of minimization}
Constrained descent methods of minimization are also known as feasible descent direction methods. The main idea is that at a given point $\mathbf{x} \in S \subset \mathbb{R}^n$, we generate a feasible direction $\Delta\mathbf{x} \in \mathbb{R}^n$ where the objective function value $f(\mathbf{x})$ can be reduced. Here $S$ is the set of feasible points, i.e., the set of points that satisfy the problem constraints. We then use a line search to set $\mathbf{x} \gets \mathbf{x} + t_\star \Delta\mathbf{x}$ for some optimal step value $t_\star > 0$, and repeat the procedure until some convergence criterion is reached.

Recall that a direction $\Delta\mathbf{x} \in \mathbb{R}^n$ is a {\em feasible descent direction} at $\mathbf{x} \in S$ if there exists $\tilde{t} > 0$ such that,
\begin{equation} \label{eq: feasible descent direction}
f(\mathbf{x} + t \Delta\mathbf{x}) < f(\mathbf{x}), \text{ and } \mathbf{x} + t \Delta\mathbf{x} \in S \text{ for all } t \in (0, \tilde{t}\,].
\end{equation}
With these ideas in mind, the general constrained descent algorithm can be stated as,

\begin{algorithm}
  \caption{General constrained descent method.
    \label{alg: general constrained descent method}}
\begin{algorithmic}[0] 
\Require{A function $f : dom(f) \subseteq \mathbb{R}^n \rightarrow \mathbb{R}$, a constrain set $S \subset dom(f)$, and a starting point $\mathbf{x} \in S$.}
\Ensure{Local/global minimum of $f{\restriction_S}$.}
	\Statex
	\Repeat
		\State {\em Direction}. Find $\Delta\mathbf{x} \in \mathbb{R}^n$ such that for some $\tilde{t} > 0$ it holds that for all $0 < t \leq \tilde{t}$, $\mathbf{x} + t \Delta\mathbf{x} \in S$ and $f(\mathbf{x} + t \Delta\mathbf{x}) < f(\mathbf{x})$. 
		\State {\em Line search}. Find an optimal step $t_\star > 0$ such that $\mathbf{x} + t_\star \Delta\mathbf{x} \in S$.
		\State {\em Update}. $\mathbf{x} \gets \mathbf{x} + t_\star \Delta\mathbf{x}$.
	\Until{stopping criterion is satisfied.}
	\State \Return{$\mathbf{x}$}
\end{algorithmic}
\end{algorithm}

\subsubsection{Projected gradient descent method}
In the case of \cref{eq: equality constrained minimization}, \cref{alg: general constrained descent method} above gives rise to  the projected gradient descent method. Here, we have that the feasible set is $S = \{\mathbf{x} \in \mathbb{R}^n : \mathbf{A} \mathbf{x} = \mathbf{y} \}$. Given a feasible point $\mathbf{x} \in S$, we know that the direction of steepest descent is $-\nabla f(\mathbf{x})$, which may not be feasible. However, its projection into $S$ will be a feasible descent direction, provided it is not zero.
\begin{definition}[Projection]We say that a matrix $\Pi \in \mathbb{R}^{n \times n}$ is a projection if and only if $\Pi^2 = \Pi$.
\end{definition}
If $\mathbf{x} \in S$, then $\Delta\mathbf{x} \in \mathbb{R}^n$ is a feasible direction at $\mathbf{x}$ if and only if $\mathbf{A} \Delta\mathbf{x} = \mathbf{0}$, in other words, if $\Delta\mathbf{x} \in \ker(\mathbf{A}) := \{ \mathbf{v} \in \mathbb{R}^n : \mathbf{A} \mathbf{v} = \mathbf{0} \}$. To see this, simply note that $\mathbf{A}(\mathbf{x} + t \Delta\mathbf{x}) = \mathbf{A}\mathbf{x} + t \mathbf{A} \Delta\mathbf{x} = \mathbf{A}\mathbf{x} = \mathbf{b}$, for any $t \in \mathbb{R}$.

\begin{definition}[Local minimizer]
Let $f : \mathbb{R}^n \rightarrow \mathbb{R}$ be a continuously differentiable function and $S$ a set of feasible solutions. A vector $\mathbf{x}^\star \in \mathbb{R}^n$ is a local minimizer of $f$ if and only if for any feasible direction $\Delta\mathbf{x} \in \mathbb{R}^n$, with $\|\Delta\mathbf{x}\|_2 = 1$, there exists a number $\tilde{t} > 0$ such that
\begin{equation} \label{eq: local minimizer}
f(\mathbf{x}^\star + t \Delta\mathbf{x}) \geq f(\mathbf{x}^\star), \quad \text{for } 0 < t \leq \tilde{t}.
\end{equation}
\end{definition}

Given a feasible point $\mathbf{x}$, we would like to find a feasible direction $\Delta\mathbf{x}$ at $\mathbf{x}$ such that $f$ improves the most, i.e., we would like to find the solution to,
\begin{equation} \label{eq: optimal feasible descent direction}
\min_{\Delta\mathbf{x} \in \ker(\mathbf{A})} \| -\!\nabla f(\mathbf{x}) - \Delta\mathbf{x}\|_2.
\end{equation}
Denote by $\mathbf{Z}$ a matrix whose columns span $\ker(\mathbf{A})$. In that case, it is easy to see that all feasible directions $\Delta\mathbf{x}$ are of the form,
\begin{equation*}
\Delta\mathbf{x} = \mathbf{Z}\mathbf{q},
\end{equation*}
for some $\mathbf{q} \in \mathbb{R}^{n-m}$. Such matrix $\mathbf{Z}$ can be formed by using, for example, an LQ-decomposition of $\mathbf{A}$, or equivalently, a QR-decomposition of $\mathbf{A}^\text{T}$. In this latter case, we can write,
\begin{equation*} \label{eq: QR-decomposition of A transpose}
\mathbf{A}^\text{T} = \mathbf{Q}\mathbf{R},
\end{equation*}
where $\mathbf{Q} \in \mathbb{R}^{n \times n}$ is a unitary matrix, and $\mathbf{R} \in \mathbb{R}^{n \times m}$ is an upper trapezoidal matrix. It is easy to see that $\mathbf{Z} = (\mathbf{q}_i)_{i=m+1}^n$, where $\mathbf{q}_i$ represents the $i$th column vector of $\mathbf{Q}$. The projection of vector $\Delta\mathbf{x}$ into the null space of $\mathbf{A}$ can then be written as,
\begin{align*}
\Pi_\mathbf{Z} \Delta\mathbf{x} &= \sum_{i=m+1}^n \langle \Delta\mathbf{x},\mathbf{q}_i \rangle\, \mathbf{q}_i , \\ 
&= \sum_{i=m+1}^{n} (\mathbf{q}_i^\text{T} \Delta\mathbf{x})\, \mathbf{q}_i, \\
&= \mathbf{Z} \mathbf{Z}^\text{T} \Delta\mathbf{x}.
\end{align*}
We define, then, the projected gradient descent direction as,
\begin{equation} \label{eq: projected gradient descent direction}
\Delta \mathbf{x}_{pg} :=
-\Pi_\mathbf{Z} \nabla f(\mathbf{x}).
\end{equation}
It is easy to see that $\Delta\mathbf{x}_{pg}$ is a feasible descent direction that solves \cref{eq: optimal feasible descent direction}. We can also show, cf., \cite{BoyVan2004}, that $\Delta\mathbf{x}_{pg} = \mathbf{v}$, where $\mathbf{v} \in \mathbb{R}^n$ and $\mathbf{w} \in \mathbb{R}^m$ solve uniquely,
\begin{equation*}
\begin{pmatrix}
\mathbf{I} & \mathbf{A}^\text{T} \\
\mathbf{A} & \mathbf{0}
\end{pmatrix}
\begin{pmatrix}
\mathbf{v} \\
\mathbf{w}
\end{pmatrix}
=
\begin{pmatrix}
-\nabla f(\mathbf{x}) \\
\mathbf{0}
\end{pmatrix}.
\end{equation*}
Observe that $\mathbf{v}$ and $\mathbf{w}$ exist because $\mathbf{A}$ is a full rank matrix.

With $\Delta\mathbf{x}_{pg}$ as the feasible descent direction, we adapt \cref{alg: general constrained descent method} into the method shown in \cref{alg: projected gradient descent method for linear constraint}.

\begin{algorithm}
  \caption{Projected gradient descent method for a linear constraint. See \cref{eq: equality constrained minimization}.
    \label{alg: projected gradient descent method for linear constraint}}
\begin{algorithmic}[0] 
\Require{A function $f : dom(f) \subseteq \mathbb{R}^n \rightarrow \mathbb{R}$, a matrix $\mathbf{A} \in \mathbb{R}^{m \times n}$ with $rank(\mathbf{A}) = m < n$, a starting point $\mathbf{x}$ such that $\mathbf{A}\mathbf{x} = \mathbf{y}$, and two parameters $\alpha \in \left(0,\frac{1}{2}\right)$ and $\beta \in (0,1)$.}
\Ensure{Local/global minimum of $f{\restriction_S}$, where $S = \{ \mathbf{v} \in \mathbb{R}^n : \mathbf{A}\mathbf{v} = \mathbf{y}\}$.}
	\Statex
	\State Compute $\mathbf{Q}\mathbf{R} = \mathbf{A}^\text{T}$, and form $\mathbf{Z} = (\mathbf{q}_i)_{i=m+1}^n$, where $\mathbf{q}_i$ is the $i$th column of $\mathbf{Q}$.
	\Repeat
		\State {\em Direction.} $\Delta\mathbf{x}_{pg} \gets \mathbf{Z}\mathbf{Z}^\text{T}\nabla f(\mathbf{x})$.
		\Comment{Do $\mathbf{Z}\left(\mathbf{Z}^\text{T}\nabla f(\mathbf{x})\right)$ in $O(n^2)$ vs $(\mathbf{Z} \mathbf{Z}^\text{T})\nabla f(\mathbf{x})$ in $O(n^3)$.}
		\State {\em Line search}. $t_\star \gets \textsc{Step\_size}(f, \mathbf{x}, \Delta\mathbf{x}_{pg},\alpha,\beta)$.
		\Comment{See \cref{alg: backtracking line search}.}
		\State {\em Update}. $\mathbf{x} \gets \mathbf{x} + t_\star \Delta\mathbf{x}_{pg}$.
	\Until{stopping criterion is satisfied.}
	\State \Return{$\mathbf{x}$}
\end{algorithmic}
\end{algorithm}

\begin{algorithm}
  \caption{Backtracking line search, see \cite{BoyVan2004}.
    \label{alg: backtracking line search}}
\begin{algorithmic}[0] 
\Require{A descent direction $\Delta\mathbf{x}$ for $f$ at $\mathbf{x} \in dom(f)$, and two parameters $\alpha \in \left(0,\frac{1}{2}\right)$ and $\beta \in (0,1)$.}
\Ensure{Step size $t$ such that $f(\mathbf{x} + t\Delta\mathbf{x}) < f(\mathbf{x})$.}
	\Statex
	\Function{Step\_size}{$f, \mathbf{x}, \Delta\mathbf{x}, \alpha, \beta$}
	\State $t \gets 1$
	\While{$f(\mathbf{x} + t\Delta\mathbf{x}) \geq f(\mathbf{x}) + \alpha t \nabla f(\mathbf{x})^\text{T}\Delta\mathbf{x}$} \Comment{$\alpha$ typically between 0.01 and 0.3}
		\State $t \gets \beta t$ \Comment{$\beta$ typically between 0.1 and 0.8}
	\EndWhile
	\State \Return{$t$}
	\EndFunction
\end{algorithmic}
\end{algorithm}


\section*{Acknowledgements}
I would like to thank Professor John J.\ Benedetto at the University of Maryland, College Park (UMD), and Dr.\ Mark Magsino, a Ph.D. student of his during the preparation of this manuscript, for their insightful and helpful comments. I would also like to thank the Institute for Physical Science and Technology at UMD for giving me the freedom to pursue this research on top of my day-to-day obligations. I want to acknowledge as well Dr. David Bowen, from the Laboratory for Physical Sciences at UMD, for the productive discussions and demonstrations that we had with him. The preparation of this manuscript is based upon work supported by the U. S. Army Research Office under grant number W911NF-17-1-0014.

\bibliographystyle{siamplain}
\bibliography{JBbib}

\begin{thebibliography}{10}

\bibitem{AhmNatRao1974}
{\sc N.~Ahmed, T.~Natarajan, and K.~R. Rao}, {\em Discrete cosine transform},
  IEEE Transactions on Computers, C-23 (1974), pp.~90--93,
  \url{https://doi.org/10.1109/T-C.1974.223784}.

\bibitem{aust2008}
{\sc D.~Austin}, {\em What is... {JPEG}?}, Notices of the AMS, 55 (2008),
  pp.~226--229.

\bibitem{Bay1976}
{\sc B.~E. Bayer}, {\em Color imaging array}, July 1976.
\newblock US Patent 3,971,065.

\bibitem{benn1941}
{\sc W.~R. Bennett}, {\em Time division multiplex systems}, The Bell System
  Technical Journal, 20 (1941), pp.~199--221.

\bibitem{BoyVan2004}
{\sc S.~P. Boyd and L.~Vandenberghe}, {\em Convex Optimization}, Cambridge
  University Press, 7th~ed., 2004 (2009).

\bibitem{BruDonEla2009}
{\sc A.~M. Bruckstein, D.~L. Donoho, and M.~Elad}, {\em From sparse solutions
  of systems of equations to sparse modeling of signals and images}, SIAM
  Review, 51 (2009), pp.~34 -- 81.

\bibitem{cand2008}
{\sc E.~J. Cand\`es}, {\em The restricted isometry property and its
  implications for compressed sensing}, Comptes Rendus Mathematique, 346
  (2008), pp.~589--592,
  \url{http://www.sciencedirect.com/science/article/pii/S1631073X08000964}.

\bibitem{CanDon2002}
{\sc E.~J. Cand\`es and D.~L. Donoho}, {\em New tight frames of curvelets and
  optimal representations of objects with piecewise-${C}^2$ singularities},
  Communications on Pure and Applied Mathematics, 57 (2004), pp.~219--266.

\bibitem{CanRomTao2006}
{\sc E.~J. Cand\`{e}s, J.~K. Romberg, and T.~Tao}, {\em Robust uncertainty
  principles: Exact signal reconstruction from highly incomplete frequency
  information}, IEEE Transactions on Information Theory, 52 (2006),
  pp.~489--509.

\bibitem{CanTao2005}
{\sc E.~J. Cand\`{e}s and T.~Tao}, {\em Decoding by linear programming}, IEEE
  Transactions on Information Theory, 51 (2005), pp.~4203--4215.

\bibitem{CheDonSau1998}
{\sc S.~S. Chen, D.~L. Donoho, and M.~A. Saunders}, {\em Atomic decomposition
  by basis pursuit}, SIAM Journal on Scientific Computing, 20 (1998),
  pp.~33--61.

\bibitem{CohDeVPetXu1999}
{\sc A.~Cohen, R.~DeVore, P.~Petrushev, and H.~Xu}, {\em Nonlinear
  approximation and the space ${BV}(\mathbb{R}^2)$}, American Journal of
  Mathematics, 121 (1999), pp.~587 -- 628.

\bibitem{dono2006}
{\sc D.~L. Donoho}, {\em Compressed sensing}, IEEE Transactions on Information
  Theory, 52 (2006), pp.~1289--1306.

\bibitem{DuDaTaLaSuKeBa2008}
{\sc M.~F. Duarte, M.~A. Davenport, D.~Takhar, J.~N. Laska, T.~Sun, K.~F.
  Kelly, and R.~G. Baraniuk}, {\em Single-pixel imaging via compressive
  sampling}, IEEE Signal Processing Magazine, 25 (2008), pp.~83 -- 91.

\bibitem{elad2010}
{\sc M.~Elad}, {\em Sparse and {R}edundant {R}epresentations: From {T}heory to
  {A}pplications in {S}ignal and {I}mage {P}rocessing}, Springer, New York,
  2010.

\bibitem{gabo1946}
{\sc D.~Gabor}, {\em Theory of communication}, Journal of the Institution of
  Electrical Engineers - Part III: Radio and Communication Engineering, 93
  (1946), pp.~429--441.
\newblock Part I - The analysis of information.

\bibitem{GelYag1957}
{\sc I.~M. Gel'fand and A.~M. Yaglom}, {\em Calculation of the amount of
  information about a random function contained in another such function},
  American Mathematical Society Translations, 2 (1957), pp.~199--246.

\bibitem{GriNasSof2008}
{\sc I.~Griva, S.~G. Nash, and A.~Sofer}, {\em Linear and Nonlinear
  Optimization}, Society for Industrial and Applied Mathematics, 2nd~ed., 2008
  (2009).

\bibitem{Kap1997}
{\sc M.~C. Kaplan}, {\em Border treatment in image processing algorithms},
  November 1997.
\newblock US Patent 5,687,258.

\bibitem{Mack2009}
{\sc D.~Mackenzie}, {\em Compressed Sensing Makes Every Pixel Count}, vol.~7 of
  What's Happening in the Mathematical Sciences, American Mathematical Society,
  2009, pp.~114 -- 127.

\bibitem{MalZha1993}
{\sc S.~G. Mallat and Z.~Zhang}, {\em Matching pursuits with time-frequency
  dictionaries}, IEEE Transactions on Signal Processing, 41 (1993), pp.~3397 --
  3415.

\bibitem{NasSof1996}
{\sc S.~G. Nash and A.~Sofer}, {\em Linear and Nonlinear Programming},
  McGraw-Hill, 1996.

\bibitem{nata1995}
{\sc B.~K. Natarajan}, {\em Sparse approximate solutions to linear systems},
  SIAM Journal on Computing, 24 (1995), pp.~227 -- 234.

\bibitem{nyqu1928}
{\sc H.~Nyquist}, {\em Certain topics in telegraph transmission theory},
  Transactions of the American Institute of Electrical Engineers, 47 (1928),
  pp.~617--644.

\bibitem{PatRezKri1993}
{\sc Y.~Pati, R.~Rezaiifar, and P.~Krishnaprasad}, {\em Orthogonal matching
  pursuit: recursive function approximation with application to wavelet
  decomposition}, in 27th Asilomar Conference on Signals, Systems and
  Computers, 1993, 1993, pp.~40--44.

\bibitem{PlaVen2000}
{\sc K.~N. Plataniotis and A.~N. Venetsanopoulos}, {\em Color Image Processing
  and Applications}, Springer Verlag, 2000.
\newblock ISBN 3-540-66953-1.

\bibitem{PonLukZelEgiCarBat2009}
{\sc N.~Ponomarenko, V.~Lukin, A.~Zelensky, K.~Egiazarian, M.~Carli, and
  F.~Battisti}, {\em {TID2008} - {A} database for evaluation of full-reference
  visual quality assessment metrics}, Advances of Modern Radioelectronics, 10
  (2009), pp.~30--45.

\bibitem{RaoYip1990}
{\sc K.~R. Rao and P.~Yip}, {\em Discrete Cosine Transform: Algorithms,
  Advantages, Applications}, Academic Press Professional, Inc., San Diego, CA,
  USA, 1990.

\bibitem{reed1993}
{\sc T.~R. Reed}, {\em Local frequency representations for image sequence
  processing and coding}, in Digital Images and Human Vision, A.~B. Watson,
  ed., The MIT Press, Cambridge, MA, 1993, pp.~3--12.

\bibitem{shan1948}
{\sc C.~E. Shannon}, {\em A mathematical theory of communication}, The Bell
  System Technical Journal, 27 (1948), pp.~379--423, 623--656.

\bibitem{shan1949}
{\sc C.~E. Shannon}, {\em Communication in the presence of noise}, Proceedings
  of the IRE, 37 (1949), pp.~10--21.

\bibitem{TauMar2002}
{\sc D.~S. Taubman and M.~W. Marcellin}, {\em JPEG 2000: Image Compression
  Fundamentals, Standards and Practice}, Kluwer Academic Publishers, Norwell,
  MA, second~ed., 2002.

\bibitem{WanBovSheSim2004}
{\sc Z.~Wang, A.~C. Bovik, H.~R. Sheikh, and E.~P. Simoncelli}, {\em Image
  quality assessment: From error measurement to structural similarity}, IEEE
  Transactions on Image Processing, 13 (2004), pp.~1--14.

\bibitem{whit1935}
{\sc J.~M. Whittaker}, {\em Interpolatory function theory}, Cambridge
  University Press, 1935.

\end{thebibliography}

 \end{document}